\documentclass[runningheads]{llncs}

\usepackage{silence}
\newcommand{\ourtitle}{Recovering Explanations from Transformed Rule-Based Ontologies}

\newcommand{\paperdoi}{10.1007/XXX}

\newif\iffullversion
\fullversiontrue
\newcommand{\fullorconf}[2]{\iffullversion#1\else#2\fi}

\newcommand{\confonly}[1]{\iffullversion\else#1\fi}

\usepackage[UKenglish]{babel}
\usepackage[utf8]{inputenc}
\usepackage[T1]{fontenc}
\usepackage{lmodern}

\usepackage{graphicx}
\usepackage{csquotes}
\usepackage{booktabs}
\usepackage{placeins}
\usepackage{xspace}
\usepackage{tikz}
\usetikzlibrary{arrows.meta, positioning}

\usepackage[hidelinks,pdftitle=\ourtitle]{hyperref}
\usepackage{color}

\usepackage{needspace}
\usepackage{mathtools}
\usepackage{amsmath,amssymb}
\usepackage{orcidlink}
\usepackage{fontawesome5}

\input{macros}

\renewcommand{\orcidID}[1]{\,\orcidlink{#1}}
\makeatletter
\newcommand\blfootnote[1]{{\def\@thefnmark{}\@footnotetext{#1}}}
\makeatother

\iffullversion
  \pdfpagewidth=\paperwidth
  \pdfpageheight=\paperheight
\fi

\begin{document}

\iffullversion
  \title{\ourtitle\thanks{This is the technical report accompanying our
      ISWC'26 paper~\protect\cite{IKM:ISWC-26}.}}
\else
  \title{\ourtitle}
\fi

\author{%
  Alex Ivliev\textsuperscript{(\faEnvelope[regular])}\orcidID{0000-0002-1604-6308} \and
  Markus Krötzsch\orcidID{0000-0002-9172-2601} \and
  Maximilian Marx\orcidID{0000-0003-1479-0341}
}
\authorrunning{A.~Ivliev et al.}
\institute{%
  Knowledge-Based Systems Group, TU Dresden, Dresden, Germany \\ \email{\{alex.ivliev, markus.kroetzsch, maximilian.marx\}@tu-dresden.de}
}

\maketitle              

\confonly{%
  \blfootnote{%
    \textcopyright{} The Author(s), under exclusive license to Springer Nature Switzerland AG 2026\\
    Editors et al. (Eds.): ISWC 2026, LNCS XXXXX, pp. 1--20, 2026.\\
    \url{https://doi.org/\paperdoi}
  }
}
\begin{abstract}
  Datalog rules are often used to define ontologies over Knowledge Graphs.
  Rule reasoners routinely optimise such ontologies
  by rewriting their rules into a form that can be evaluated more efficiently.
  These transformations preserve the entailed facts,
  but not the structure of the underlying derivations.
  A proof tree under the rewritten rules explains why a fact holds,
  but does not readily yield an explanation in terms of the original rules.
  We study the problem of constructing,
  from a proof of entailment under the rewritten rules,
  a proof under the original ones:
  we establish its computational complexity
  and identify two practically relevant languages for specifying proof transformations.


\end{abstract}

\section{Introduction}

Datalog and its numerous extensions form a family of rule languages
that can express ontological knowledge in the Semantic Web
context~\cite{KT16:BreakingKGRules,BSG:Vadalog18,GIKM2022}. 
OWL~RL and OWL~EL, two profiles of the Web Ontology Language OWL~2~\cite{owl2-profiles}, 
can be captured by 
Datalog rules~\cite{DBLP:journals/semweb/UrbaniPHB14,Kroetzsch10:elreason,Kroetzsch11:elreason},
and extensions of Datalog cover even larger subsets of
OWL~2~\cite{CDK18:combined-hornalchoiq,CDK:RuleReasoning:KI20}.
Datalog is also used for query answering~\cite{CaliGL12:datalogplusontologies,DBLP:conf/semweb/CarralDK17},
for encoding ontological reasoning~\cite{Kroetzsch11:elreason}, 
for data extraction and analysis~\cite{Piro+:RDFoxKaiserPermanente16,DIADEM14}, 
and for transforming Knowledge Graphs~\cite{SkvortsovXBL25,DGHIKMM2026}.
A key benefit of Datalog as an ontology language is the inherent explainability of its consequences: 
any inferred fact can be justified by a \emph{proof tree} 
that shows which rules were applied to which facts, 
tracing every inference back to the input.
Computing such proof trees is well-studied for both
OWL~\cite{KPHS:owl-justifications07,HPS08:lacjust,ABBDKM22:evonne,DBLP:conf/lpar/AlrabbaaBBKK20,ABKK22:exomq,DBLP:journals/cgf/MendezAKLBD23}
and
Datalog~\cite{EKM:exruleprov:rr2022,DBLP:journals/toplas/ZhaoSS20,Green+:ProvSemiRings07,DBLP:journals/pacmmod/CalauttiLPS24}.
For Datalog, it usually takes the form of why-provenance, 
which asks for not one, 
but \emph{all} proof trees of a fact; 
selecting the \enquote*{most helpful} among them is a research topic in its own
right~\cite{A+:elproofshapes:dl25,DBLP:conf/dlog/BorgwardtHKW20}.
Even when proof trees grow very large, 
summarisation and visualisation make them accessible to
users~\cite{MGWIKR26:nev-eval,DGHIKMM2026,DBLP:journals/cgf/MendezAKLBD23}.
Another major strength of Datalog is 
the availability of highly optimised engines that support efficient reasoning 
over knowledge graphs with billions of 
facts~\cite{N+15:RDFoxToolPaper,UJK:VLog2016,Jordan+:Souffle16,Ivliev+:Nemo2024,google-logica-docs}.
This efficiency is partly achieved 
by rewriting the input rules into an equivalent, optimised ontology, 
using techniques such as program minimisation~\cite{Sagiv88}, 
the Magic Sets rewriting for goal-directed evaluation~\cite{BancilhonMSU86}, 
removal of redundant parameters~\cite{RamakrishnanBK88}, 
folding and unfolding of rules~\cite{PettorossiP94}, and propagation of filter conditions
through recursive rules~\cite{HK2026}. 
Some of these techniques were even discovered by studying transformations on proof
trees~\cite{DBLP:journals/jcss/RamakrishnanSUV93,CV92:equivalence}.
Such transformations preserve the consequences of an ontology,
but may substantially change the structure of its derivations,
as demonstrated in the following example:

\begin{example}\label{ex:transitivity-intro}
  The transitive closure $r$ of a binary relation~$e$ can be
  defined by two equivalent ontologies $\aprogram_1$ and $\aprogram_2$:

  \medskip\noindent
  \begin{minipage}[t]{0.46\textwidth}
  \textbf{Ontology $\aprogram_1$ (bushy):}
  \begin{align}
      r(x, y) &\leftarrow e(x, y) \label{rule:intro-bushy-base} \\
      r(x, y) &\leftarrow r(x, z), r(z, y) \label{rule:intro-bushy-step}
  \end{align}
  \end{minipage}%
  \hfill
  \begin{minipage}[t]{0.46\textwidth}
  \textbf{Ontology $\aprogram_2$ (left-linear):}
  \begin{align}
      r(x, y) &\leftarrow e(x, y) \label{rule:intro-lin-base} \\
      r(x, y) &\leftarrow r(x, z), e(z, y) \label{rule:intro-lin-step}
  \end{align}
  \end{minipage}

  \medskip
  Both ontologies derive the same $r$-facts on any input database.
  However,
  the bushy rule~\eqref{rule:intro-bushy-step} 
  allows splitting a path at any intermediate node,
  which may result in balanced proof trees of logarithmic depth.
  The left-linear rule~\eqref{rule:intro-lin-step}
  extends a path by one edge at a time,
  producing left-leaning chains.
  Figure~\ref{fig:intro-transformation}
  shows proof trees for $r(a,d)$
  over $\Dnter = \{e(a, b), e(b, c), e(c, d)\}$.
\end{example}
A Datalog engine may rewrite one of these ontologies into the other to evaluate it efficiently. 
Proof trees are then computed over the rewritten ontology, 
and thus fail to explain why a fact follows under the \emph{original} one. 
We study the problem of recovering such an explanation, 
i.e., of constructing a corresponding proof tree for the original ontology 
from a proof tree for the transformed ontology. 
We establish the computational complexity of this task (Section~\ref{sec:hardness}) 
and, building on tree transducers~\cite{Baker78b} 
and Monadic Second-Order-definable transductions~\cite{Courcelle94,EngelfrietM99}, 
identify two suitable formalisms for specifying proof tree transformations (Section~\ref{sec:languages}).
\fullorconf{The appendix of this paper}{An extended technical report\footnote{\url{https://iccl.inf.tu-dresden.de/web/Inproceedings3474/en}}} contains further examples.

\begin{figure}[t]
\centering
\begin{tikzpicture}[>=Stealth,
    nd/.style={draw, rounded corners, font=\footnotesize, inner sep=3pt},
    ed/.style={fill=gray!15},
    rl/.style={font=\scriptsize, text=black!50, anchor=west, inner sep=1pt, xshift=1pt},
  ]
  \node[nd,ed] (L1) at (0,0) {$e(a,b)$};
  \node[nd] (L2) at (0,1) {$r(a,b)$};
  \node[rl] at (L2.east) {\eqref{rule:intro-bushy-base}};
  \node[nd,ed] (L3) at (1.8,0) {$e(b,c)$};
  \node[nd] (L4) at (1.8,1) {$r(b,c)$};
  \node[rl] at (L4.east) {\eqref{rule:intro-bushy-base}};
  \node[nd,ed] (L5) at (3.6,0) {$e(c,d)$};
  \node[nd] (L6) at (3.6,1) {$r(c,d)$};
  \node[rl] at (L6.east) {\eqref{rule:intro-bushy-base}};
  \node[nd] (L7) at (2.7,2) {$r(b,d)$};
  \node[rl] at (L7.east) {\eqref{rule:intro-bushy-step}};
  \node[nd] (L8) at (1.35,3) {$r(a,d)$};
  \node[rl] at (L8.east) {\eqref{rule:intro-bushy-step}};
  \draw[->] (L1) -- (L2);
  \draw[->] (L3) -- (L4);
  \draw[->] (L5) -- (L6);
  \draw[->] (L4) -- (L7);
  \draw[->] (L6) -- (L7);
  \draw[->] (L2) -- (L8);
  \draw[->] (L7) -- (L8);
\end{tikzpicture}%
\hfil
\begin{tikzpicture}[>=Stealth,
    nd/.style={draw, rounded corners, font=\footnotesize, inner sep=3pt},
    ed/.style={fill=gray!15},
    rl/.style={font=\scriptsize, text=black!50, anchor=west, inner sep=1pt, xshift=1pt},
  ]
  \node[nd,ed] (R1) at (0,0) {$e(a,b)$};
  \node[nd] (R2) at (0,1) {$r(a,b)$};
  \node[rl] at (R2.east) {\eqref{rule:intro-lin-base}};
  \node[nd,ed] (R3) at (2.0,1) {$e(b,c)$};
  \node[nd] (R4) at (1.0,2) {$r(a,c)$};
  \node[rl] at (R4.east) {\eqref{rule:intro-lin-step}};
  \node[nd,ed] (R5) at (3.0,2) {$e(c,d)$};
  \node[nd] (R6) at (2.0,3) {$r(a,d)$};
  \node[rl] at (R6.east) {\eqref{rule:intro-lin-step}};
  \draw[->] (R1) -- (R2);
  \draw[->] (R2) -- (R4);
  \draw[->] (R3) -- (R4);
  \draw[->] (R4) -- (R6);
  \draw[->] (R5) -- (R6);
\end{tikzpicture}
\caption{Proof trees for ontologies $\aprogram_1$ (left)
and $\aprogram_2$ (right) in Example~\ref{ex:transitivity-intro}. }
\label{fig:intro-transformation}
\end{figure}
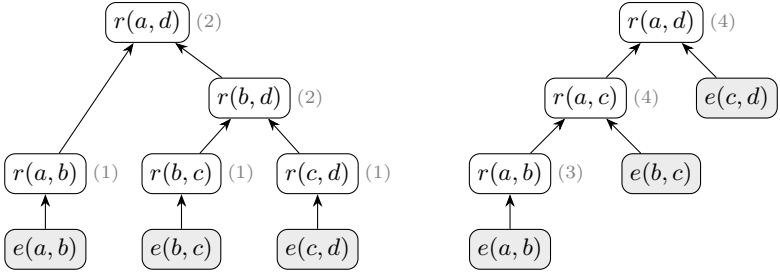

\Needspace*{5\baselineskip}
\section{Preliminaries}\label{sec:preliminaries}

We recall the syntax of Datalog (with built-in relations)
and define its semantics based on proof trees.
A recent survey~\cite{Kroetzsch2025:Datalog} offers a broader treatment.

\subsubsection{Syntax}

We consider a vocabulary of mutually disjoint sets $\Plang$ of \newterm{predicates},
$\Clang$ of \newterm{constants}, and $\Vlang$ of \newterm{variables}.
We associate an \newterm{arity} $\arity(p) \geq 0$ with every predicate $p \in \Plang$.
A \newterm{term} is an element of the set $\Tlang \defeq \Clang \cup \Vlang$.
We use $\vec{t}$ to denote a list of terms $\tuple{t_1, \ldots, t_k}$
where $\sizeof{\vec{t}} = k$. 
Similarly, we write, e.g., $\vec{x}$ or $\vec{y}$ for lists of variables. 
An \newterm{atom} has the form $p(\vec{t})$ where $p \in \Plang$ and $\arity(p) = \sizeof{\vec{t}}$.
Atoms that do not contain any variables are called \newterm{facts}.
The sets of all facts and all atoms are denoted $\Flang$ and $\Alang$, respectively.
A \newterm{Datalog rule} $\arule$ is an expression of the form 
\begin{equation}
    H \leftarrow B_1, \ldots, B_n
\end{equation}
where $H$ and $B_1, \ldots, B_n$ are atoms.
We call $H$ the \newterm{head} and $B_1, \ldots, B_n$
the \newterm{body} of $\arule$.
We use notation $\rhead{\arule} \defeq H$ and
$\rbodyi{i}{\arule} \defeq B_i$ for $1 \leq i \leq n$,
and set $\sizeof{\arule} \defeq n$.  
We require that every variable in the head of $\arule$ 
also appears in its body (\newterm{safety}).
Rule bodies are allowed to be empty; in this case, the symbol $\leftarrow$ is omitted.
A \newterm{program} $\aprogram$ is a triple $\tuple{R, \Pinput, \Poutput}$
consisting of a finite set $R$ of rules,
and sets $\Pinput$ and $\Poutput$ of \newterm{input} and \newterm{output predicates},
respectively, such that input predicates may not appear in the head of a rule.

\subsubsection{Semantics and Proofs}

A (potentially infinite) set of facts $\Dnter$ is a \newterm{database}.
It is an \newterm{input database} for $\aprogram = \tuple{R, \Pinput, \Poutput}$
if the predicate of every fact is from $\Pinput$. 
We allow for infinite databases $\Dnter$ to express conceptually infinite built-in relations
such as $\leq$ and $\neq$. An \newterm{ontology} $\mathcal{O} = \tuple{\aprogram, \Dnter}$ consists of a program $\aprogram$ together with an input database $\Dnter$.

An \newterm{assignment} $\mu$ for a rule $\arule$ 
maps every variable $x$ in $\arule$ to a constant $\mu(x) \in \Clang$.
We extend $\mu$ to terms by setting $\mu(c) \defeq c$ for constants $c \in \Clang$,
to lists of terms by $\mu(\tuple{t_1, \ldots, t_k}) \defeq \tuple{\mu(t_1), \ldots, \mu(t_k)}$,
and to atoms by $\mu(p(\vec{t})) \defeq p(\mu(\vec{t}))$.

An \newterm{ordered tree} is a tuple $T = \tuple{V, \troot, \ch}$ 
where $V$ is a finite set of \newterm{nodes}, 
$\troot \in V$ is called the \newterm{root} node, 
and $\ch\colon V \to V^*$ maps each node $v$ to a finite sequence 
$\ch(v) = \tuple{v_1, \ldots, v_n}$ of distinct \newterm{children} of $v$, 
such that every node except $\troot$ appears in exactly one such sequence 
and $\troot$ appears in none.
We call $v$ a \newterm{leaf} if $\ch(v)$ is empty.
A \newterm{proof tree} for a fact $f$ over a program $\aprogram = \tuple{R, \Pinput, \Poutput}$
and a database $\Dnter$ is an ordered tree $T = \tuple{V, \troot, \ch}$ 
together with labelling functions $\tfact \colon V \to \Flang$ 
and $\trule \colon V \to R \cup \set{\mleaf}$, such that 
\begin{enumerate}
  \item $\tfact(\troot) = f$;
  \item if $\ch(v) = \tuple{v_1, \ldots, v_n}$ with $n \geq 1$, 
  then $\trule(v) \in R$ with $\sizeof{\trule(v)} = n$, 
  and there exists an assignment $\mu$ for $\trule(v)$ such that 
  $\mu(\rhead{\trule(v)}) = \tfact(v)$ and 
  $\mu(\rbodyi{i}{\trule(v)}) = \tfact(v_i)$ for $1 \leq i \leq n$; and
  \item for every leaf $v$, either:
  \begin{enumerate}
      \item $\trule(v) = \mleaf$ and $\tfact(v) \in \Dnter$, or  
      \item $\trule(v) \in R$ with $\sizeof{\trule(v)} = 0$ and $\tfact(v) = \rhead{\trule(v)}$.
  \end{enumerate}
\end{enumerate}
We say that a fact $f$ is \newterm{derived} from $\aprogram$ and $\Dnter$,
written $\aprogram, \Dnter \vdash f$,
if there exists a proof tree $T$ for $f$ over $\aprogram$ and $\Dnter$.
The \newterm{leaf database} $\leafdb{T}$ of a proof tree $T$ is the set of all
facts $\tfact(v)$ for leaves $v$ of $T$ with $\trule(v) = \mleaf$.
Every proof tree $T$ for $f$ over $\aprogram$ and $\Dnter$ is also a proof tree
for $f$ over $\aprogram$ and $\leafdb{T}$.
The \newterm{output} of $\aprogram$ over $\Dnter$ is the set
$\pout(\aprogram, \Dnter) \defeq \set{ f \mid \aprogram, \Dnter \vdash f \text{ and the predicate of } f \text{ is in } \Poutput }$.

\begin{remark}
For finite input databases, $\pout(\aprogram, \Dnter)$ coincides 
with the set of output facts in the least fixpoint of the standard 
immediate-consequence operator.
With infinite built-in relations, practical systems ensure termination by,
e.g., requiring that every variable in a built-in
also occurs in a body atom with a finite predicate.
Such considerations do not affect our results.
\end{remark}

\subsubsection{Containment}
Let $\aprogram_1 = \tuple{R_1, \Pinput, \Poutput}$ and 
$\aprogram_2 = \tuple{R_2, \Pinput, \Poutput}$ be programs
over the same input and output predicates.
We say that $\aprogram_1$ is \newterm{contained} in $\aprogram_2$,
written $\aprogram_1 \contained \aprogram_2$,
if $\pout(\aprogram_1, \Dnter) \subseteq \pout(\aprogram_2, \Dnter)$ 
for every input database $\Dnter$.
The two programs are \newterm{equivalent}, written $\aprogram_1 \equiv \aprogram_2$,
if $\aprogram_1 \contained \aprogram_2$ and $\aprogram_2 \contained \aprogram_1$.
A \newterm{program transformation} is a function $t$ mapping programs $\aprogram$ to programs $t(\aprogram)$.
It is \newterm{sound} if $t(\aprogram) \contained \aprogram$ 
and \newterm{complete} if $\aprogram \contained t(\aprogram)$ for all programs $\aprogram$.

\section{Proof Transformations}\label{sec:transformations}

A sound program transformation produces
a new program $\aprogram_1$ from an original program $\aprogram_2$,
such that $\aprogram_1 \contained \aprogram_2$.
This guarantees that for every proof tree $T_1$ over $\aprogram_1$ 
there exists a proof tree $T_2$ over $\aprogram_2$ deriving the same (output) fact
over the same input database.
The central question of this paper is how to obtain $T_2$ from $T_1$,
and which computational resources this requires.

In a proof tree, the same fact may be derived multiple times 
in independent subtrees.
A more compact representation merges such redundant derivations,
yielding a directed acyclic graph.
A \newterm{proof DAG} for a fact $f$ over a program $\aprogram$ and database $\Dnter$
is a tuple $G = \tuple{V, \troot, \ch, \tfact, \trule}$
satisfying the same conditions as a proof tree,
except that a node may appear as a child of multiple nodes.
Notions such as the leaf database $\leafdb{G}$ carry over from proof trees.
The \newterm{unfolding} of a proof DAG $G$ is the proof tree $\unfold(G)$
obtained by recursively duplicating shared nodes.
Conversely, a proof tree may be compressed into a proof DAG by merging nodes labelled with the 
same fact and the same rule and whose subgraphs are isomorphic.
For a fact $f = p(t_1, \ldots, t_k)$, we define $\sizeof{f} \defeq k$.
The size of a proof tree $T = \tuple{V, \troot, \ch, \tfact, \trule}$ is
$\sizeof{T} \defeq \sum_{v \in V} \sizeof{\tfact(v)}$.
The facts of a proof DAG $G$ are $\factsof{G} \defeq \bigcup_{v \in V} \tfact(v)$, and its size is $\sizeof{G} \defeq \sizeof{\set{\tuple{v, w} \in V \times V \mid v \in \ch(w)}} + \sum_{t \in \factsof{G}} \sizeof{t}$.

\begin{definition}[Proof transformation]\label{def:transformation}
  Let $\aprogram_1 = \tuple{R_1, \Pinput, \Poutput}$ and $\aprogram_2 = \tuple{R_2, \Pinput, \Poutput}$
  be programs with $\aprogram_1 \contained \aprogram_2$.
  A \newterm{proof transformation} from $\aprogram_1$ to $\aprogram_2$ 
  is a function $\tau$ that maps 
  every proof tree $T_1$ for a fact $f$ with predicate $p \in \Poutput$
  over $\aprogram_1$ and input database $\Dnter$ 
  to a proof DAG $\tau(T_1)$ for $f$ over $\aprogram_2$ and $\Dnter$.
  The \newterm{result} of this transformation is the proof tree $\unfold(\tau(T_1))$.
\end{definition}
We require that $\tau(T_{1})$ is a proof DAG instead of a proof tree 
because proof trees can be exponentially larger than the corresponding proof DAGs. 
The following example demonstrates such an unavoidable blowup.

\begin{example}\label{ex:dag-blowup}
Consider the following equivalent programs $\aprogram_1$ and $\aprogram_2$
with input predicates $\Pinput = \set{s, e}$ and output predicate $\Poutput = \set{r}$.

\medskip
\noindent
\begin{minipage}[t]{0.42\textwidth}
\textbf{Program $\aprogram_1$}:
\begin{align}
  r(x) &\leftarrow s(x) \label{rule:exp-1-base} \\
  r(y) &\leftarrow r(x), e(x, y) \label{rule:exp-1-step}
\end{align}
\end{minipage}%
\hfill
\begin{minipage}[t]{0.54\textwidth}
\textbf{Program $\aprogram_2$}:
\begin{align}
  r(x) &\leftarrow s(x) \label{rule:exp-2-base} \\
  r(y) &\leftarrow p(x), q(x), e(x, y) \label{rule:exp-2-step} \\
  p(x) &\leftarrow r(x) \label{rule:exp-2-p} \\
  q(x) &\leftarrow r(x) \label{rule:exp-2-q}
\end{align}
\end{minipage}

\medskip\noindent
Both programs derive $r(x)$ for exactly those $x$ 
reachable from some $s$-fact along $e$-edges.
However, $\aprogram_2$ duplicates each derivation of $r(x)$  
through the auxiliary predicates $p$ and $q$.
Consider $\Dnter = \set{s(c), e(c, b), e(b, a)}$.
A proof transformation $\tau$ from $\aprogram_1$ to $\aprogram_2$
must map nodes for rules \eqref{rule:exp-1-base} and \eqref{rule:exp-1-step}
to nodes for rules \eqref{rule:exp-2-base} and \eqref{rule:exp-2-step}, respectively,
and introduce two nodes for rules \eqref{rule:exp-2-p} and \eqref{rule:exp-2-q}
below each node for rule \eqref{rule:exp-2-step}.
Figure~\ref{fig:dag-blowup} shows the proof tree $T_{1}$ for $r(a)$ over $\aprogram_{1}$ and $\Dnter$,
and the proof DAG $\tau(T_{1})$ over $\aprogram_{2}$ and $\Dnter$.
The unfolding $\unfold(\tau(T_{1}))$ is exponential in the size of $T_{1}$.
\end{example}

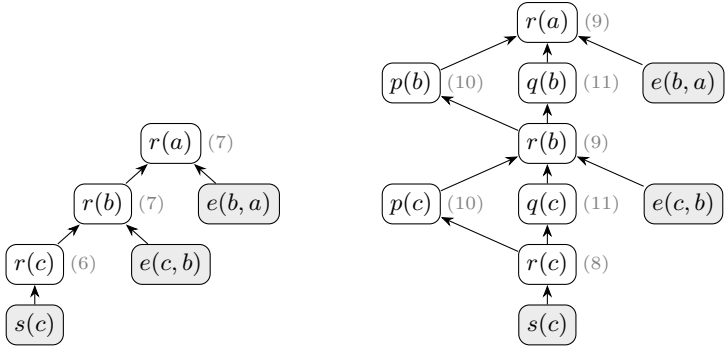
\begin{figure}[t]
\centering
\begin{tikzpicture}[>=Stealth,
    nd/.style={draw, rounded corners, font=\footnotesize, inner sep=3pt},
    ed/.style={fill=gray!15},
    rl/.style={font=\scriptsize, text=black!50, anchor=west, inner sep=1pt, xshift=1pt},
  ]
  \node[nd,ed] (L1) at (0,0) {$s(c)$};
  \node[nd] (L2) at (0,0.8) {$r(c)$};
  \node[rl] at (L2.east) {\eqref{rule:exp-1-base}};
  \node[nd,ed] (L3) at (1.8,0.8) {$e(c,b)$};
  \node[nd] (L4) at (0.9,1.6) {$r(b)$};
  \node[rl] at (L4.east) {\eqref{rule:exp-1-step}};
  \node[nd,ed] (L5) at (2.7,1.6) {$e(b,a)$};
  \node[nd] (L6) at (1.8,2.4) {$r(a)$};
  \node[rl] at (L6.east) {\eqref{rule:exp-1-step}};
  \draw[->] (L1) -- (L2);
  \draw[->] (L2) -- (L4);
  \draw[->] (L3) -- (L4);
  \draw[->] (L4) -- (L6);
  \draw[->] (L5) -- (L6);
\end{tikzpicture}%
\hfil
\begin{tikzpicture}[>=Stealth,
    nd/.style={draw, rounded corners, font=\footnotesize, inner sep=3pt},
    ed/.style={fill=gray!15},
    rl/.style={font=\scriptsize, text=black!50, anchor=west, inner sep=1pt, xshift=1pt},
  ]
  \node[nd,ed] (N1) at (2,0) {$s(c)$};
  \node[nd] (N2) at (2,0.8) {$r(c)$};
  \node[rl] at (N2.east) {\eqref{rule:exp-2-base}};
  \node[nd] (N3) at (0.2,1.6) {$p(c)$};
  \node[rl] at (N3.east) {\eqref{rule:exp-2-p}};
  \node[nd] (N4) at (2,1.6) {$q(c)$};
  \node[rl] at (N4.east) {\eqref{rule:exp-2-q}};
  \node[nd,ed] (N5) at (3.8,1.6) {$e(c,b)$};
  \node[nd] (N6) at (2,2.4) {$r(b)$};
  \node[rl] at (N6.east) {\eqref{rule:exp-2-step}};
  \node[nd] (N7) at (0.2,3.2) {$p(b)$};
  \node[rl] at (N7.east) {\eqref{rule:exp-2-p}};
  \node[nd] (N8) at (2,3.2) {$q(b)$};
  \node[rl] at (N8.east) {\eqref{rule:exp-2-q}};
  \node[nd,ed] (N9) at (3.8,3.2) {$e(b,a)$};
  \node[nd] (N10) at (2,4) {$r(a)$};
  \node[rl] at (N10.east) {\eqref{rule:exp-2-step}};
  \draw[->] (N1) -- (N2);
  \draw[->] (N2) -- (N3);
  \draw[->] (N2) -- (N4);
  \draw[->] (N3) -- (N6);
  \draw[->] (N4) -- (N6);
  \draw[->] (N5) -- (N6);
  \draw[->] (N6) -- (N7);
  \draw[->] (N6) -- (N8);
  \draw[->] (N7) -- (N10);
  \draw[->] (N8) -- (N10);
  \draw[->] (N9) -- (N10);
\end{tikzpicture}
\caption{Proof tree $T_1$ for $\aprogram_1$ (left) 
and proof DAG $\tau(T_1)$ for $\aprogram_2$ (right)
of Example~\ref{ex:dag-blowup}}
\label{fig:dag-blowup}
\end{figure}

\section{Complexity of Proof Transformations}\label{sec:hardness}
We now focus on the computational problem of transforming 
a proof tree $T_1$ over $\aprogram_1$ 
into a proof DAG $T_2$ over $\aprogram_2$ for the same fact. 
A naive approach extracts the leaf database of $T_1$ and evaluates $\aprogram_2$ from scratch,
which takes time polynomial in the size of $T_1$.
In this section, we show that producing a valid output is \PTime-hard
if $T_1$ is encoded as a directed acyclic graph, 
matching the naive upper bound, 
and \LogCFL-hard if $T_1$ is encoded as a tree.

\subsection{The Transformation Problem}
A proof transformation (Definition~\ref{def:transformation})
assigns a fixed output proof to every input proof.
Here, instead, we study the computational complexity of constructing \emph{some} valid proof over $\aprogram_2$
from a proof over $\aprogram_1$ for the same fact, 
accepting any valid proof DAG as output.
Intuitively, this captures the complexity of the best possible transformation algorithm:
lower bounds apply to every proof transformation, 
while the upper bounds are established by a particular transformation algorithm.
Formally, we define the search problem as follows.

\smallskip
\noindent
\fbox{\parbox{0.96\textwidth}{%
    \textbf{Problem:} $\TransEnc(\aprogram_1, \aprogram_2)$ for programs $\aprogram_1 \contained \aprogram_2$

    \noindent\textbf{Input:}
    A proof tree (if $\enc = \enctree$)
    or proof DAG (if $\enc = \encdag$)
    $T_1$ for a fact $f$ whose predicate is in $\Poutput$,
    over $\aprogram_1$ and input database $\leafdb{T_1}$.

    \noindent\textbf{Output:}
    A proof DAG $T_2$ for $f$ over $\aprogram_2$ and $\leafdb{T_1}$.
  }}
\smallskip

Note that we fix the programs $\aprogram_1 \contained \aprogram_2$ in the problem definition,
so complexity is measured only in terms of the size of the input.
Consequently, algorithms for $\TransEnc(\aprogram_1, \aprogram_2)$
may exploit properties of the particular program pair.
We consider both encodings $\enc \in \set{\enctree, \encdag}$ of the input proof tree
(both given as adjacency lists),
because they may differ exponentially in size (see Example~\ref{ex:dag-blowup}).
This difference stems solely from sharing identical subderivations, 
making the tree encoding the more informative measure when sharing is limited.

A search problem associates every \newterm{instance} with a
(non-empty) set of \newterm{solutions}.  
An algorithm \newterm{solves} such a problem if, 
for every instance, it computes a valid solution.
To establish hardness, we use the following notion of reduction,
adapted from the standard treatment of multivalued function problems
in complexity theory~\cite{DBLP:journals/tcs/MegiddoP91}: 
An \NCone \newterm{reduction} from a decision problem $L$ to a search problem
$S$ is a pair of functions $g$ and $h$, 
computable by logspace-uniform \NCone circuit families~\cite{greenlaw1995limits}, 
such that for every instance $x$ of $L$: 
(1) $g(x)$ is an instance of $S$, and 
(2) for every solution $y$ of $g(x)$, we have $h(y) = 1$ iff $x \in L$.
Condition~(2) quantifies over all solutions since an algorithm solving
$S$ may output any of them.  A search problem $S$ is
$\mathcal{C}$-hard under \NCone reductions for a complexity
class $\mathcal{C}$ if every decision problem in $\mathcal{C}$ reduces
to $S$ via an \NCone reduction.
We only consider \NCone reductions here, since they are sufficiently weak to establish \LogSpace-hardness, yet expressive enough to admit \LogSpace-complete problems to reduce from.

\subsection{Complexity Results}
The naive approach gives us an upper bound for the transformation problem:
\begin{proposition}\label{prop:trans-upper}
  Let $\aprogram_1 \contained \aprogram_2$ be programs. 
  Then there exists a polynomial-time algorithm that solves $\TransEnc(\aprogram_1, \aprogram_2)$ 
  for both encodings $\enc \in \set{\enctree, \encdag}$.
\end{proposition}
\begin{proof}
  Given $T_1$, we compute the least fixpoint of $\aprogram_2$ over
  $\leafdb{T_1}$ in polynomial time~\cite{Alice}, 
  storing for each fact the rule instance 
  that was used to derive it. 
  We get $f \in \pout(\aprogram_2, \leafdb{T_1})$ since $\aprogram_1 \contained \aprogram_2$.
  Then, the stored rule instances, starting from $f$, 
  form a proof DAG $T_2$ for $f$ over $\aprogram_2$ and $\leafdb{T_1}$. 
  \qed
\end{proof}
For the corresponding lower bound,
we reduce from the Circuit Value Problem, which is known
to be \PTime-complete for DAG-encoded circuits~\cite{greenlaw1995limits}.  
As an immediate corollary, we also obtain a \LogSpace-hardness for
$\TransTree(\aprogram_{1}, \aprogram_{2})$, as this problem 
is \LogSpace-complete when the circuit is given as a tree~\cite{BeaudryM95}.
We show a \LogCFL lower bound later.
\begin{theorem}\label{thm:hardness-dag}
  There is a pair of programs $\aprogram_{1} \contained \aprogram_{2}$ for which the problem
  $\TransDAG(\aprogram_{1}, \aprogram_{2})$ is \PTime-hard.
\end{theorem}

\begin{proof}
  We define programs $\aprogram_1$ and $\aprogram_2$
  over the single output predicate $\predname{goal}$ and the set of 
  input predicates $\Pinput = \set{\prednamesub{t}{in}, \prednamesub{f}{in}, \predname{and}, \predname{not},     \predname{out}}$,   
  where we use facts $\predname{and}(g, g_1, g_2)$ and $\predname{not}(g, g_1)$ to encode the structure of the     circuit, 
  $\prednamesub{t}{in}(g)$ and $\prednamesub{f}{in}(g)$ to declare the input gate~$g$ as true or false, 
  and $\predname{out}(g)$ to mark the output gate.
  Program $\aprogram_2$ evaluates the circuit as expected:

  \noindent
  \begin{minipage}[b]{0.48\textwidth}%
    \begin{align}
      \label{rule:t-tin}
      \predname{true}(g) &\leftarrow \prednamesub{t}{in}(g) \\
      \label{rule:f-fin}
      \predname{false}(g) &\leftarrow \prednamesub{f}{in}(g) \\
      \begin{split}
        \label{rule:a-tt}
        \predname{true}(g) &\leftarrow \predname{and}(g,g_1,g_2), \\
        &\phantom{\leftarrow~{ }} \predname{true}(g_1),\, \predname{true}(g_2)
      \end{split} \\
      \begin{split}
        \label{rule:a-ft}
        \predname{false}(g) &\leftarrow \predname{and}(g,g_1,g_2), \\
        &\phantom{\leftarrow~{ }} \predname{false}(g_1),\, \predname{true}(g_2) 
      \end{split} \\
      \label{rule:g-f}
      \predname{goal} &\leftarrow \predname{false}(g),\, \predname{out}(g)
    \end{align}
  \end{minipage}%
  \hfill
  \begin{minipage}[b]{0.48\textwidth}
    \begin{align}
      \begin{split}
        \label{rule:a-tf}
        \predname{false}(g) &\leftarrow \predname{and}(g,g_1,g_2), \\
        &\phantom{\leftarrow~{ }} \predname{true}(g_1),\, \predname{false}(g_2)
      \end{split} \\
      \begin{split}
        \label{rule:a-ff}
        \predname{false}(g) &\leftarrow \predname{and}(g,g_1,g_2), \\
        &\phantom{\leftarrow~{ }} \predname{false}(g_1),\, \predname{false}(g_2)
      \end{split} \\
      \label{rule:n-t} 
      \predname{true}(g) &\leftarrow \predname{not}(g,g_1),\,
      \predname{false}(g_1) \\
      \label{rule:n-f}      
      \predname{false}(g) &\leftarrow \predname{not}(g,g_1),\,
      \predname{true}(g_1) \\
      \label{rule:g-t}
      \predname{goal} &\leftarrow \predname{true}(g),\, \predname{out}(g)
    \end{align}
  \end{minipage}
  \medskip
  
  \noindent Proof trees of $\aprogram_2$ correspond to evaluations of the circuit encoded by the input facts, 
  where the result of the output gate is obtained by checking the rule label at the root of the proof tree. 
  When constructing $\aprogram_1$, we need to ensure that
  $\predname{goal}$ is only derived when $\aprogram_2$ can derive a
  valuation for an output gate:

  \begin{align}
    \label{rule:m-tin}
    \predname{maybe}(g) &\leftarrow \prednamesub{t}{in}(g) \\
    \label{rule:m-fin}
    \predname{maybe}(g) &\leftarrow \prednamesub{f}{in}(g) \\
    \label{rule:m-and}         
    \predname{maybe}(g) &\leftarrow \predname{maybe}(g_1),
    \predname{maybe}(g_2), 
    \predname{and}(g, g_1, g_2) \\
    \label{rule:m-not}          
    \predname{maybe}(g) &\leftarrow \predname{not}(g, g_1), \predname{maybe}(g_1) \\
    \label{rule:m-goal}
    \predname{goal} &\leftarrow \predname{maybe}(g),\, \predname{out}(g) 
  \end{align}
  Proof trees of $\aprogram_1$ encode the structure of the circuit,
  but do not evaluate it. 
  We can thus obtain a suitable input DAG $T_{1}$ 
  from an arbitrary circuit by a simple relabelling that maps
  the gate types to the corresponding rules in $\aprogram_1$.
  \qed
\end{proof}
We now show the \LogCFL lower bound for $\TransTree(\aprogram_{1}, \aprogram_{2})$.
As in the proof of Theorem~\ref{thm:hardness-dag},
program $\aprogram_{1}$ merely validates the structure of the input,
while $\aprogram_{2}$ decides a concrete problem.
Here, we reduce from the word problem of a fixed context-free language.
\begin{theorem}\label{thm:hardness-tree}
  There is a pair of programs
  $\aprogram_{1} \contained \aprogram_{2}$ for which the problem
  $\TransTree(\aprogram_{1}, \aprogram_{2})$ is \LogCFL-hard.
\end{theorem}
\begin{proof}
  Let $L$ be a context-free language 
  with grammar $G = \tuple{V, \Sigma, P, S}$ and $\Sigma \subseteq \Clang$.
  For this reduction, 
  we pick $L$ to be a language whose word problem is \LogCFL-complete under \NCone reductions~\cite{DBLP:journals/jcss/LautemannMSV01}
  and assume $G$ to be given in quadratic Greibach normal form~\cite{DBLP:reference/hfl/AutebertBB97},
  i.e., each production $p \in P$ is of the form $A \to \sigma A_{1} \dotsb A_{n}$
  for $0 \leq n \leq 2$, a terminal symbol $\sigma \in \Sigma$, and nonterminal
  symbols $A, A_{1}, \dotsc, A_{n} \in V$.
  We use input predicates
  $\Pinput = \set{\predname{first}, \predname{next}, \predname{last}}$
  and the output predicate $\predname{goal}$, and encode a word
  $w = \sigma_{1}\dotsb \sigma_{\ell} \in \Sigma^{+}$ as the database
  \begin{align*}
    \Dnter_{w} \defeq \set{ \predname{first}(c_{1}),
    \predname{next}(c_{1}, \sigma_{1}, c_{2}), \dotsc,
    \predname{next}(c_{\ell}, \sigma_{\ell}, c_{\ell + 1}),
    \predname{last}(c_{\ell + 1}) },
  \end{align*}
  where $\predname{next}(x, \sigma, y)$ expresses that position $y$
  follows position $x$ when reading the letter $\sigma$.
  Program $\aprogram_{1}$ consists of the following rules, where rule~\eqref{eq:lcfl-p1-next} is instantiated for every $\sigma \in \Sigma$:
  \begin{align}
    \predname{r}(x, y) &\leftarrow \predname{next}(x, \sigma, y) \label{eq:lcfl-p1-next} \\
    \predname{r}(x, y) &\leftarrow \predname{r}(x, z), \predname{r}(z, y) \label{eq:lcfl-p1-transitive} \\
    \predname{goal} &\leftarrow \predname{r}(x, y), \predname{first}(x), \predname{last}(y) \label{eq:lcfl-p1-goal}
  \end{align}
  Program $\aprogram_2$ includes the rules of form \eqref{eq:lcfl-p1-next} and rule \eqref{eq:lcfl-p1-transitive}.
  For each nonterminal $A \in V$, 
  it uses binary predicates $\predname{A}$ and $\bar{\predname{A}}$,
  intended to hold on a span of the word iff $A$ can,
  respectively cannot, produce that span.
  Program $\aprogram_2$ further contains the following two rules,
  which can be used to decide if $w \in L$:
  \begin{align}
    \label{eq:lcfl-p2-goal} \predname{goal} &\leftarrow
    \predname{first}(x), \predname{last}(y), \predname{S}(x, y) \\
    \label{eq:lcfl-p2-goal-nonmatch} \predname{goal} &\leftarrow
    \predname{first}(x), \predname{last}(y), \bar{\predname{S}}(x, y)
  \end{align}
  The predicates $\predname{A}$ are derived by the rules given below,
  which are instantiated for each production $A \to \sigma A_{1} \dotsb A_{n} \in P$:
  \begin{align}
    \label{eq:lcfl-p2-production} \predname{A}(x, x_{n + 1}) &\leftarrow
    \predname{next}(x, \sigma, x_{1}), \predname{A}_{1}(x_{1}, x_{2}), \dotsc, \predname{A}_{n}(x_{n}, x_{n + 1})
  \end{align}
  Deriving the complements $\bar{\predname{A}}$ is more challenging,
  as $\aprogram_2$ must verify that \emph{no} production of $A$ can produce the span.
  For every production $p \in P$, we use binary predicate $\bar{\predname{p}}$
  that holds on every span that cannot be produced by $p$.
  Then, $\aprogram_2$ contains for each nonterminal $A \in V$ derived 
  by productions $\predname{p}^1, \ldots, \predname{p}^m \in P$:
  \begin{align}
    \label{eq:lcfl-p2-nonterminal-nonmatch}
    \bar{\predname{A}}(x, y) &\leftarrow
    \bar{\predname{p}}^{1}(x, y), \dotsc, \bar{\predname{p}}^{m}(x, y)
  \end{align}
  A production $p$ with leading terminal $\sigma$ fails on every span
  that does not start with $\sigma$.
  We derive these spans using an auxiliary predicate $\bar\sigma$:
  \begin{align}
    \label{eq:lcfl-p2-terminal-nonmatch} \bar{\sigma}(x, y) &\leftarrow
    \predname{next}(x, \sigma', y) \\
    \label{eq:lcfl-p2-production-terminal-nonmatch} \bar{\predname{p}}(x, y) &\leftarrow
    \bar{\sigma}(x, y) \\
    \label{eq:lcfl-p2-production-terminal-nonmatch-tail} \bar{\predname{p}}(x, y) &\leftarrow
    \bar{\sigma}(x, z), \predname{r}(z, y)
  \end{align}
  for all $\sigma \neq \sigma' \in \Sigma$ and productions $p \in P$ with lead terminal $\sigma$.
  The remaining causes of failure depend on the shape of $p$.
  A production $p = A \to \sigma$ additionally fails on every span of length at least two:
  \begin{align}
    \label{eq:lcfl-p2-production-nonterminal-long}
    \bar{\predname{p}}(x, y) &\leftarrow
    \predname{next}(x, t, z), \predname{r}(z, y)
  \end{align}
  A production $p = A \to \sigma A_{1}$ additionally fails on spans of
  length one, which are too short, and on spans whose remainder after
  $\sigma$ cannot be produced by $A_{1}$:
  \begin{align}
    \label{eq:lcfl-p2-production-nonterminal-short}
    \bar{\predname{p}}(x, y) &\leftarrow \predname{next}(x, t, y) \\
    \label{eq:lcfl-p2-production-nonterminal-nonmatch-1}
    \bar{\predname{p}}(x, y) &\leftarrow
    \predname{next}(x, \sigma, z), \bar{\predname{A}}_{1}(z, y)
  \end{align}
  Finally, consider a production $p = A \to \sigma A_{1} A_{2}$.
  Spans of length one or two are again too short:
  \begin{align}
    \bar{\predname{p}}(x, y) &\leftarrow \predname{next}(x, t, y)
    \label{eq:lcfl-p2-production-nonterminal-nonmatch-2-shortest} \\
    \bar{\predname{p}}(x, y) &\leftarrow
    \predname{next}(x, t, z), \predname{next}(z, t', y)
    \label{eq:lcfl-p2-production-nonterminal-nonmatch-2-short}
  \end{align}
  On longer spans,
  $p$ fails iff \emph{every} way of splitting the word between $A_{1}$ and $A_{2}$ fails.
  We capture this condition with two auxiliary predicates.
  The first predicate $\bar{\predname{p}}_{\predname{s}}(z, s, y)$ 
  states that the span from $z$ to $y$ cannot be split at position $s$,
  since $A_{1}$ fails on the part before $s$ or $A_{2}$ fails on the part after it.
  The second predicate $\bar{\predname{p}}_{\leq \predname{s}}(z, s, y)$ states that
  all split positions up to $s$ fail:
  \begin{align}
    \bar{\predname{p}}_{\predname{s}}(z, s, y) &\leftarrow
    \bar{\predname{A}}_{1}(z, s), \predname{r}(s, y)
    \label{eq:lcfl-p2-production-nonterminal-nonmatch-2-split-base-left} \\
    \bar{\predname{p}}_{\predname{s}}(z, s, y) &\leftarrow
    \predname{r}(z, s), \bar{\predname{A}}_{2}(s, y)
    \label{eq:lcfl-p2-production-nonterminal-nonmatch-2-split-base-right} \\
    \bar{\predname{p}}_{\leq \predname{s}}(z, s, y) &\leftarrow
    \predname{next}(z, t, s), \bar{\predname{p}}_{\predname{s}}(z, s, y)
    \label{eq:lcfl-p2-production-nonterminal-nonmatch-2-split-sweep-base} \\
    \bar{\predname{p}}_{\leq \predname{s}}(z, s', y) &\leftarrow
    \bar{\predname{p}}_{\leq \predname{s}}(z, s, y),
    \predname{next}(s, t, s'), \bar{\predname{p}}_{\predname{s}}(z, s', y)
    \label{eq:lcfl-p2-production-nonterminal-nonmatch-2-split-sweep-step} \\
    \bar{\predname{p}}(x, y) &\leftarrow
    \predname{next}(x, \sigma, z),
    \bar{\predname{p}}_{\leq \predname{s}}(z, s, y),
    \predname{next}(s, t, y)
    \label{eq:lcfl-p2-production-nonterminal-nonmatch-2-split}
  \end{align}

  For containment, consider an input database $\Dnter$ 
  with $\predname{goal} \in \pout(\aprogram_{1}, \Dnter)$.
  Then there are
  constants $c_{1}, \dotsc, c_{\ell + 1}$ and
  $s_{1}, \dotsc, s_{\ell}$ such that $\Dnter$ contains the facts
  $\predname{first}(c_{1})$, $\predname{last}(c_{\ell + 1})$, and
  $\predname{next}(c_{i}, s_{i}, c_{i + 1})$ for all
  $1 \leq i \leq \ell$, spelling out the word
  $w = s_{1}\dotsb s_{\ell}$.
  If $w \in L$, we can use rules of form \eqref{eq:lcfl-p2-production}
  to simulate productions of $G$,
  deriving $A(c_i, c_{j+1})$ if the nonterminal $A$ can produce the 
  span $s_i \dotsb s_j$.
  Hence, we have
  $\aprogram_{2}, \Dnter \vdash \predname{S}(c_{1}, c_{\ell + 1})$,
  and thus $\aprogram_{2}, \Dnter \vdash \predname{goal}$ by
  rule~\eqref{eq:lcfl-p2-goal}.

  For the case $w \notin L$ we show, 
  by induction on the length of a span $c_{i}, \dotsc, c_{j + 1}$, 
  that $\aprogram_{2}$ derives $\bar{\predname{A}}(c_{i}, c_{j + 1})$ 
  for every nonterminal $A \in V$ that does not produce $s_{i}\dotsb s_{j}$.
  It suffices by rule~\eqref{eq:lcfl-p2-nonterminal-nonmatch} to derive
  $\bar{\predname{p}}(c_{i}, c_{j + 1})$ for every production
  $p = A \to \sigma A_{1}\dotsb A_{n} \in P$ that cannot produce
  $s_{i}\dotsb s_{j}$.
  If $s_{i} \neq \sigma$, this follows from
  rules~\eqref{eq:lcfl-p2-production-terminal-nonmatch} and~\eqref{eq:lcfl-p2-production-terminal-nonmatch-tail}.
  Otherwise, we distinguish the three shapes of $p$.
  For $n = 0$, the production fails only if the span is longer than
  one letter, which is covered by
  rule~\eqref{eq:lcfl-p2-production-nonterminal-long}.
   For $n = 1$, it fails if the
  span consists of a single letter (rule~\eqref{eq:lcfl-p2-production-nonterminal-short}),
  or if $A_{1}$ does not produce the remainder $s_{i + 1}\dotsb s_{j}$,
  in which case the induction hypothesis yields
  $\bar{\predname{A}}_{1}(c_{i + 1}, c_{j + 1})$ and
  rule~\eqref{eq:lcfl-p2-production-nonterminal-nonmatch-1} applies.
  For $n = 2$, spans of one
  or two letters are too short (rules~\eqref{eq:lcfl-p2-production-nonterminal-nonmatch-2-shortest}
  and~\eqref{eq:lcfl-p2-production-nonterminal-nonmatch-2-short}); 
  on longer spans, every split of the remainder must fail. 
  By the induction hypothesis,
  rules~\eqref{eq:lcfl-p2-production-nonterminal-nonmatch-2-split-base-left}
  and~\eqref{eq:lcfl-p2-production-nonterminal-nonmatch-2-split-base-right} then derive
  $\bar{\predname{p}}_{\predname{s}}(c_{i + 1}, c_{k}, c_{j + 1})$ for
  every split position $c_{k}$, so that 
  rules~\eqref{eq:lcfl-p2-production-nonterminal-nonmatch-2-split-sweep-base}
  and~\eqref{eq:lcfl-p2-production-nonterminal-nonmatch-2-split-sweep-step} 
  derive $\bar{\predname{p}}_{\leq\predname{s}}(c_{i + 1}, c_j, c_{j + 1})$
  and therefore 
  rule~\eqref{eq:lcfl-p2-production-nonterminal-nonmatch-2-split} applies.
  We obtain
  $\aprogram_{2}, \Dnter \vdash \bar{\predname{S}}(c_{1}, c_{\ell + 1})$ for the full span,
  and hence $\predname{goal} \in \pout(\aprogram_{2}, \Dnter)$ by
  rule~\eqref{eq:lcfl-p2-goal-nonmatch}.

  On the encoding $\Dnter_{w}$ of a word $w$
  exactly one of $\predname{S}(c_{1}, c_{\ell + 1})$ and
  $\bar{\predname{S}}(c_{1}, c_{\ell + 1})$ is derivable. 
  Hence the rule at the root of every proof tree for $\predname{goal}$ over
  $\aprogram_{2}$ and $\Dnter_{w}$ is determined by whether $w \in L$.
  \qed
\end{proof}

\section{Proof Transformation Languages}\label{sec:languages}

We now turn to the question of \emph{how} to specify proof transformations 
and identify two practical transformation languages.
In Section~\ref{sec:homomorphism}, 
we show that the notion of uniform containment between programs admits local transformations 
defined by proof tree homomorphisms. 
We consider Monadic Second-Order Logic (MSO) interpretations 
as a transformation language in Section~\ref{sec:mso}, 
which define transformations computable in logarithmic space for tree-encoded inputs.

\subsection{Proof Tree Homomorphisms}\label{sec:homomorphism}

A simple class of proof transformation arises from \emph{uniform containment},
a notion introduced by Sagiv for minimising Datalog programs~\cite{Sagiv88}.
In this setting, an atom or a rule is
considered redundant if removing it from the program results in an equivalent program.
Checking containment between two Datalog programs is undecidable~\cite{Shmueli87},
so uniform containment can be used as a decidable alternative.

\begin{definition}[Uniform containment]\label{def:uniform}
  Let $\aprogram_1 = \tuple{R_1, \Pinput, \Poutput}$ 
  and $\aprogram_2 = \tuple{R_2, \Pinput, \Poutput}$ be programs
  over the same input and output predicates.
  Program $\aprogram_1$ is \newterm{uniformly contained} in $\aprogram_2$, denoted $\aprogram_1 \ucontained \aprogram_2$,
  if for every database $\Dnter$ (not only input databases for $\aprogram_1$ or $\aprogram_2$),
  we have $\pout(\aprogram_1, \Dnter) \subseteq \pout(\aprogram_2, \Dnter)$.
\end{definition}
Note that containment compares the output of the programs on input databases,
while for uniform containment $\pout(\aprogram_1, \Dnter) \subseteq \pout(\aprogram_2, \Dnter)$ 
must hold over all databases $\Dnter$.
Hence, $\aprogram_1 \ucontained \aprogram_2$ implies $\aprogram_1 \contained \aprogram_2$,
but the converse does not hold in general.
Let $\omega \colon \Vlang \to \Clang$ be an injective mapping from variables to constants.
We can decide whether $\aprogram_1 \ucontained \aprogram_2$ by verifying for every rule $\rho_1 \in R_1$,
that its head can be derived from its body using $\aprogram_2$, 
i.e., that $\aprogram_2, \Dnter \vdash \omega(\rhead{\rho_1})$ 
for $\Dnter = \set{\omega(\rbodyi{i}{\rho_1}) \mid 1 \leq i \leq |\rho_1|}$~\cite{Sagiv88}.
This leads to a class of proof transformations 
that replace each node of a proof tree for $\aprogram_1$ with the matching derivation of the same fact in $\aprogram_2$.
This corresponds naturally to the notion of tree homomorphisms~\cite{tree1984},
which we adapt for proof trees as follows.

\begin{definition}[Proof tree template]\label{def:proof-tree-template}
  Let $\aprogram = \tuple{R, \Pinput, \Poutput}$ be a program and $\rho$ be a rule.
  A \newterm{proof tree template} $T$ for $\rho$ over $\aprogram$
  is a tuple $\tuple{V, \troot, \ch, \tatom, \trule}$ 
  where $\tuple{V, \troot, \ch}$ is an ordered tree
  with labelling functions $\tatom \colon V \to \Alang$ 
  and $\trule \colon V \to R$, such that
  \begin{enumerate}
    \item $\tatom(\troot) = \rhead{\rho}$;
    \item if $\ch(v) = \tuple{v_1, \ldots, v_n}$ with $n \geq 1$, 
          then $\trule(v) \in R$ with $\sizeof{\trule(v)} = n$,
          and there exists a substitution 
          $\theta_v \colon \Vlang \to \Tlang$ such that
          $\theta_v(\rhead{\trule(v)}) = \tatom(v)$ 
          and $\theta_v(\rbodyi{i}{\trule(v)}) = \tatom(v_i)$ 
          for $1 \leq i \leq n$; and
    \item for every leaf node $\ell$, there is an index $\rho_\ell$ 
          with $1 \leq \rho_\ell \leq \sizeof{\rho}$
          such that $\tatom(\ell) = \rbodyi{\rho_\ell}{\rho}$.
  \end{enumerate}
\end{definition}
Let $T$ be a proof tree template for a rule $\rho$ with $\sizeof{\rho} = n$
over $\aprogram$ and $L \subseteq V$
the set of leaves of $T$.
Furthermore, let $T_1, \ldots, T_n$ be proof DAGs over $\aprogram$
with $T_i = \tuple{V_i, \troot_i, \ch_i, \tfact_i, \trule_i}$
that do not share nodes with $T$ or with each other.
If there is a substitution $\sigma \colon \Vlang \to \Clang$
with $\sigma(\rbodyi{i}{\rho}) = \tfact(\troot_i)$ for $1 \leq i \leq n$,
we write $T[T_1, \ldots, T_n]$ to denote the proof DAG $S$
that results from uniformly replacing each node
in $\tuple{V', \troot', \ch', \tfact', \trule'}$ with a fresh copy, where
\begin{itemize}
    \item $V' \defeq (V \cup \bigcup_{i = 1}^n V_i) \setminus L$;
    \item $\troot' \defeq \troot$;
    \item $\ch' \defeq \ch \cup \bigcup_{i = 1}^n \ch_i$ replacing every occurrence of a leaf node $\ell \in L$ by $\troot_{\rho_\ell}$;
    \item $\tfact'(v) \defeq \sigma(\tatom(v))$ for $v \in V$ and $\tfact'(v) \defeq \tfact_i(v)$ for $v \in V_i$; and
    \item $\trule'(v) \defeq \trule(v)$ for $v \in V$ and $\trule'(v) \defeq \trule_i(v)$ for $v \in V_i$.
\end{itemize}
Further, we use $T_{\mid v}$ to denote the subgraph of a proof DAG $T$ rooted at $v$.
We can now define proof tree homomorphisms as follows.
\begin{definition}[Proof tree homomorphism]\label{def:proof-tree-homomorphism}
    Let $\aprogram_1 = \tuple{R_1, \Pinput, \Poutput}$ 
    and $\aprogram_2 = \tuple{R_2, \Pinput, \Poutput}$ 
    be programs with $\aprogram_1 \contained \aprogram_2$,
    and let $\set{S_\rho}_{\rho \in R_1}$ 
    be a family of proof tree templates
    where $S_\rho$ is a proof tree template for $\rho \in R_1$ over $\aprogram_2$.
    The \newterm{proof tree homomorphism}
    induced by $\set{S_\rho}_{\rho \in R_1}$ 
    is a proof transformation $\tau$ defined as follows.
    Given a proof tree $T = \tuple{V, \troot, \ch, \tfact, \trule}$ 
    over $\aprogram_1$ and some database $\Dnter$, we obtain $\tau(T)$ recursively:
    \begin{enumerate}
        \item if $v$ is a leaf node, then $\tau(T_{\mid v}) \defeq T_{\mid v}$; 
        \item otherwise, if $\ch(v) = \tuple{v_1, \ldots, v_n}$ and $\trule(v) = \rho$, then
        \begin{equation*}
            \tau(T_{\mid v}) \defeq S_\rho[\tau(T_{\mid v_1}), \ldots, \tau(T_{\mid v_n})]
        \end{equation*}
    \end{enumerate}
\end{definition}
From the above characterisation of uniform containment, we immediately obtain: 
\begin{theorem}
    Let $\aprogram_1$ and $\aprogram_2$ be programs.
    Then there exists a proof tree homomorphism from $\aprogram_1$ to $\aprogram_2$
    if and only if $\aprogram_1 \ucontained \aprogram_2$.
\end{theorem}
\begin{example}\label{ex:uniform}
    Recall the programs $\aprogram_1$ (bushy) and $\aprogram_2$ (left-linear) 
    from Example~\ref{ex:transitivity-intro}.
    We have $\aprogram_2 \ucontained \aprogram_1$,
    so we can obtain a proof tree homomorphism from $\aprogram_2$ to $\aprogram_1$.
    For the base rule~\eqref{rule:intro-lin-base}, 
    the template $T_{\eqref{rule:intro-lin-base}}$ applies the corresponding base 
    rule~\eqref{rule:intro-bushy-base} of $\aprogram_1$ directly:
    
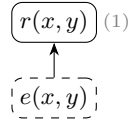
\begin{center}
    \begin{tikzpicture}[>=Stealth,
        nd/.style={draw, rounded corners, font=\footnotesize, inner sep=3pt},
        lf/.style={draw, rounded corners, font=\footnotesize, inner sep=3pt, dashed},
        rl/.style={font=\scriptsize, text=black!50, anchor=west, inner sep=1pt, xshift=1pt},
      ]
      \node[lf] (L1) at (0,0) {$e(x,y)$};
      \node[nd] (L2) at (0,1) {$r(x,y)$};
      \node[rl] at (L2.east) {\eqref{rule:intro-bushy-base}};
      \draw[->] (L1) -- (L2);
    \end{tikzpicture}
    \end{center}
    For the recursive rule~\eqref{rule:intro-lin-step},
    the template $T_{\eqref{rule:intro-lin-step}}$ uses the bushy 
    rule~\eqref{rule:intro-bushy-step} at the root
    and derives $r(z,y)$ from the second body atom $e(z,y)$ 
    via~\eqref{rule:intro-bushy-base}:
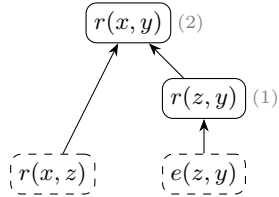
    \begin{center}
    \begin{tikzpicture}[>=Stealth,
        nd/.style={draw, rounded corners, font=\footnotesize, inner sep=3pt},
        lf/.style={draw, rounded corners, font=\footnotesize, inner sep=3pt, dashed},
        rl/.style={font=\scriptsize, text=black!50, anchor=west, inner sep=1pt, xshift=1pt},
      ]
      \node[lf] (L1) at (0,0) {$r(x,z)$};
      \node[lf] (L2) at (2,0) {$e(z,y)$};
      \node[nd] (L3) at (2,1) {$r(z,y)$};
      \node[rl] at (L3.east) {\eqref{rule:intro-bushy-base}};
      \node[nd] (L4) at (1,2) {$r(x,y)$};
      \node[rl] at (L4.east) {\eqref{rule:intro-bushy-step}};
      \draw[->] (L2) -- (L3);
      \draw[->] (L1) -- (L4);
      \draw[->] (L3) -- (L4);
    \end{tikzpicture}
    \end{center}
\end{example}
\fullorconf{Appendix~\ref{appendix:uniform}}{The technical report} contains further examples of proof tree homomorphisms.

\subsection{Transformations Based on Monadic Second-Order Logic}\label{sec:mso}

Proof tree homomorphisms are restricted to local transformations
and cannot consider the global structure of the input tree.
For instance, the reverse direction of Example~\ref{ex:uniform},
transforming bushy proof trees of $\aprogram_1$ 
into left-linear proof trees of $\aprogram_2$,
cannot be realised by proof tree homomorphisms.
MSO transductions~\cite{Courcelle94,EngelfrietM99} 
are a well-studied formalism for tree-to-tree transformations 
that can express such global restructurings.
However, the output of MSO transductions is at most linear in the input,
while proof transformations 
may require polynomial growth:

\begin{example}\label{ex:non-linear-growth}
  Consider the following programs $\aprogram_1$ and $\aprogram_2$:

  \medskip
  \noindent
  \begin{minipage}[t]{0.36\textwidth}
  \textbf{Program $\aprogram_1$:}
  \begin{align}
    r(x) &\leftarrow s(x) \label{rule:growth-1-base} \\
    r(y) &\leftarrow r(x), n(x, y) \label{rule:growth-1-step} \\
    \predname{goal} &\leftarrow r(x), e(x) \label{rule:growth-1-goal}
  \end{align}
  \end{minipage}%
  \hfill
  \begin{minipage}[t]{0.56\textwidth}
  \textbf{Program $\aprogram_2$:}
  \begin{align}
    r(x, x) &\leftarrow s(x) \label{rule:growth-2-base} \\
    r(x, y') &\leftarrow r(x, y), n(y, y') \label{rule:growth-2-step-r} \\
      r(x', y') &\leftarrow r(x, y), n(x, x'), s(y'), e(y) \label{rule:growth-2-step-l} \\
    \predname{goal} &\leftarrow r(x, y), e(x), e(y) \label{rule:growth-2-goal}
  \end{align}
  \end{minipage}
  
  \medskip
  \noindent Program $\aprogram_1$ derives $\predname{goal}$ if there is a path of $n$-edges from some $s$-fact to an $e$-fact.
  Program $\aprogram_2$ implements a counter, 
  traversing all possible pairs of values along the $n$-chain.
  Hence, proof trees for $\predname{goal}$ over $\aprogram_1$
  are of size linear in the input,
  while proof trees for $\predname{goal}$ over $\aprogram_2$
  are of size quadratic in the input.
\end{example}
We accommodate for such polynomial growth of degree $k$ by using \newterm{MSO interpretations}~\cite{Bojanczyk22},
a generalisation of MSO transductions 
in which output elements are defined as $k$-tuples of input elements.

\subsubsection{Encoding proof trees as logical structures}
The goal is to define proof transformations 
as MSO formulae.
With each proof tree, we therefore associate a relational structure
consisting of the following relations over a universe of tree nodes and constants:
\begin{itemize}
  \item $\msochild_i(x, y)$ if $x$ is the  $i$th child of $y$;
  \item $\msorule_\rho(x)$ if rule $\rho$ is applied to obtain $x$; and
  \item $\msoarg_j(x, d)$ if the $j$th argument of the fact at $x$ is $d$.
\end{itemize}
We can now express proof trees as databases over this signature.
\begin{definition}\label{def:associated-interpretation}
  Let $T = \tuple{V, \troot, \ch, \tfact, \trule}$ be a proof tree 
  over a program $\aprogram = \tuple{R, \Pinput, \Poutput}$
  and some database $\Dnter$.
  The \newterm{associated interpretation} of $T$ is the database $\msotree(T)$,
  which contains exactly the following facts:
  \begin{itemize}
    \item $\msochild_i(v, w)$ for every $v, w \in V$ such that $\ch(w) = \tuple{v_1, \ldots, v_n}$ with $v = v_i$;
    \item $\msorule_\rho(v)$ for every $v \in V$ and $\rho \in R$ such that $\trule(v) = \rho$; and
    \item $\msoarg_j(v, d)$ for every $v \in V$ such that $\tfact(v) = p(d_1, \ldots, d_m)$ with $d = d_j$.
  \end{itemize} 
\end{definition}

\subsubsection{MSO transformations}

We can now use MSO formulae to specify a proof transformation
by defining the output structure over the nodes of the input proof tree.
For an introduction to Monadic Second-Order Logic, we refer to Libkin~\cite{Libkin04}.
The following definition for MSO interpretations is adapted to our setting from Gallot~\etal~\cite{GallotLN25}.

\begin{definition}[MSO interpretation]\label{def:mso-interpretation}
  Let $\aprogram = \tuple{R, \Pinput, \Poutput}$ be a program.
  Let $n$ be the maximum arity of any rule in $R$
  and let $m$ be the maximum arity of any atom occurring in $R$.
  A \newterm{$k$-dimensional MSO interpretation} for $\aprogram$ with $k \geq 1$
  consists of 
  \begin{itemize}
    \item a finite set $\mathcal{C}$ of components,
    \item for each component $\alpha \in \mathcal{C}$, an MSO formula $\varphi_\text{dom}^\alpha(\vec{x})$ with $k$ free variables,
    \item for each $\alpha \in \mathcal{C}$ and $\rho \in R$,
          an MSO formula $\varphi_\rho^\alpha(\vec{x})$ with $k$ free variables,
    \item for each $\alpha, \beta \in \mathcal{C}$ and $1 \leq i \leq n$, 
          an MSO formula  $\varphi_{\text{child}}^{\alpha, \beta, i}(\vec{x}, \vec{y})$  
          with $2k$ free variables, and
    \item for each $\alpha \in \mathcal{C}$, and $1 \leq j \leq m$, 
          an MSO formula $\varphi_{\text{arg}}^{\alpha, j}(\vec{x}, d)$ 
          with $k + 1$ free variables.
  \end{itemize}
\end{definition}
Conceptually, an MSO interpretation defines an output tree 
whose nodes are tuples of input nodes selected by $\varphi_\text{dom}^\alpha(\vec{x})$
for each component $\alpha \in \mathcal{C}$. 
The remaining formulae define the labellings of children, rules, and facts.
We write $\Dnter, \vec{v} \models \varphi(\vec{x})$ 
to denote that the MSO formula $\varphi$ 
is satisfied in the database $\Dnter$ 
when the free variables $\vec{x}$ are assigned to the elements $\vec{v}$.

\begin{definition}[MSO proof transformation]\label{def:mso-proof-tree-transformation}
    Let $\aprogram_1 = \tuple{R_1, \Pinput, \Poutput}$ and 
    $\aprogram_2 = \tuple{R_2, \Pinput, \Poutput}$ be programs
    with $\aprogram_1 \contained \aprogram_2$.
    A proof transformation $\tau$ from $\aprogram_1$ to $\aprogram_2$
    is an \newterm{MSO proof transformation} 
    if there is a $k$-dimensional MSO interpretation for $\aprogram_2$
    and some $k \geq 1$ and some set of components $\mathcal{C}$
    such that for every proof tree $T_1 = \tuple{V_1, \troot_1, \ch_1, \tfact_1, \trule_1}$
    for some fact $f$ with predicate $p \in \Poutput$
    over $\aprogram_1$ and database $\Dnter$, there is a proof DAG 
    $G_2 = \tau(T_1) = \tuple{V_2, \troot_2, \ch_2, \tfact_2, \trule_2}$
    for $f$ over $\aprogram_2$ and $\Dnter$ where
    \begin{itemize}
      \item the nodes of $G_2$ are pairs $\tuple{\alpha, \vec{v}} \in \mathcal{C} \times V_1^k$ 
            with $\msotree(T_1), \vec{v} \models \varphi_{\text{dom}}^\alpha(\vec{x})$;
      \item the rule at node $\tuple{\alpha, \vec{v}}$ is the unique $\rho \in R_2$
            such that $\msotree(T_1), \vec{v} \models \varphi_\rho^\alpha(\vec{x})$;
      \item node $\tuple{\beta, \vec{w}}$ is the $i$-th child of $\tuple{\alpha, \vec{v}}$ 
            iff $\msotree(T_1), \vec{v}, \vec{w} \models \varphi_{\text{child}}^{\alpha, \beta, i}(\vec{x}, \vec{y})$; and
      \item the $j$-th argument of the fact at $\tuple{\alpha, \vec{v}}$ is $d$ 
            iff $\msotree(T_1), \vec{v}, d \models \varphi_{\text{arg}}^{\alpha, j}(\vec{x}, d)$.
    \end{itemize}
\end{definition}

\begin{example}\label{ex:mso-transitivity}
  The proof transformation from bushy trees to left-linear trees 
  from Example~\ref{ex:transitivity-intro} can be described 
  as a 1-dimensional MSO proof transformation with components $\mathcal{C} = \set{\alpha}$ as follows.

  We assume a relation $v \leq w$ that holds whenever $v = w$ or $w$ is an ancestor of $v$
  and further define $\predname{Leaf}(v) \defeq \neg \exists u.\; \msochild(u, v)$
  and $\predname{Root}(v) \defeq \neg \exists w.\; \msochild(v, w)$
  where $\msochild(u, v) \defeq \bigvee_i \msochild_i(u,v)$.
  We then order the nodes above the leaves (derived by rule~\eqref{rule:intro-bushy-base}) 
  of the input tree from left to right:
  \begin{align}
    \begin{split}
      {u <_\ell v} \defeq{}
          &\msorule_\eqref{rule:intro-bushy-base}(u) \land 
          \msorule_\eqref{rule:intro-bushy-base}(v) \\
          &\land \exists p, a, b.\; \msochild_1(a, p) \land \msochild_2(b, p) \land
          u \leq a \land v \leq b
    \end{split}
  \end{align}
  We can check whether two nodes $u$ and $v$ are next to each other in the order
  \begin{align}
      \predname{Cons}(u, v) \defeq{} u <_\ell v \land \neg \exists w.\; u <_\ell w \land w <_\ell v
  \end{align}
  We then identify the leftmost and rightmost nodes labelled with rule~\eqref{rule:intro-bushy-base}:
  \begin{align}
    \predname{First}(v) &\defeq{} \msorule_\eqref{rule:intro-bushy-base}(v) \land \neg \exists w.\; \predname{Cons}(w, v) \\
    \predname{Last}(v) &\defeq{} \msorule_\eqref{rule:intro-bushy-base}(v) \land \neg \exists w.\; \predname{Cons}(v, w)
  \end{align}
  With this in place, we can define the structure of the output tree as follows:
  \begin{align}
    \varphi_\text{dom}(x) &\defeq{} \neg \msorule_\eqref{rule:intro-bushy-step}(x) \\
    \varphi^{1}_\text{child}(x, y) &\defeq{} \big( \msochild_1(x, y) \land \predname{First}(y) \big) \vee \predname{Cons}(x, y) \\
    \varphi^{2}_\text{child}(x, y) &\defeq{} \neg \predname{First}(y) \land \msorule_\eqref{rule:intro-bushy-base}(y) \land \msochild_1(x, y) \\
    \varphi_\eqref{rule:intro-lin-base}(x) &\defeq{} \predname{First}(x) \\
    \varphi_\eqref{rule:intro-lin-step}(x) &\defeq{} \exists y.\; \predname{Cons}(y, x) \\
    \begin{split}
      \varphi^{1}_\text{arg}(x, d) &\defeq{} \big( \predname{Leaf}(x) \land \msoarg_1(x, d) \big) \\
      &\phantom{{}\defeq{}} \vee \big( \msorule_\eqref{rule:intro-bushy-base}(x) \land \exists r.\; \predname{Root}(r) \land \msoarg_1(r, d) \big)
    \end{split} \\
    \varphi^{2}_\text{arg}(x, d) &\defeq{} \msoarg_2(x, d)
  \end{align}
\end{example}

\subsubsection{Complexity}

We now show that MSO proof transformations can be evaluated in logarithmic space
when the input is encoded as a tree.
Our argument uses the fact that MSO model checking over structures of 
bounded treewidth can be performed in logarithmic space~\cite{ElberfeldJT10}.
We begin by recalling the relevant notions.
Given a tree $T$ with vertex set $V$ and a subset $X \subseteq V$,
we write $T[X]$ for the subtree of $T$ induced by~$X$,
that is, the tree obtained by restricting $T$ to the nodes in $X$
while preserving the child relationship among them.
For a database $\Dnter$, we write $\dbactive(\Dnter)$
for the set of all constants occurring in $\Dnter$.

\begin{definition}[Treewidth]\label{def:treewidth}
  Let $\Dnter$ be a database. 
  A \newterm{tree decomposition} of $\Dnter$ 
  is a labelled tree $T = \tuple{V, E, B}$
  with labelling function $B \colon V \to \set{X \mid X \subseteq \dbactive(\Dnter)}$
  such that
  \begin{enumerate}
      \item for all $a \in \dbactive(\Dnter)$, the induced subtree $T[X]$
            with $X = \set{v \in V \mid a \in B(v)}$ is nonempty and connected; and \label{def:treewidth-1}
          \item for every $p(t_1, \ldots, t_k) \in \Dnter$,
            there is a node $v \in V$ with $\set{t_1, \ldots, t_k} \subseteq B(v)$. \label{def:treewidth-2}
  \end{enumerate}
  The \newterm{width} of a tree decomposition is $\max_{v \in V} \sizeof{B(v)} - 1$.
  The \newterm{treewidth} of $\Dnter$ is the minimal width of any of its tree decompositions;
  the treewidth of a proof tree $T$ is the treewidth of $\msotree(T)$.
\end{definition}

The treewidth of a proof tree $T$ over $\aprogram$ might grow arbitrarily large 
due to the presence of $\msoarg_j(v, d)$ facts in $\msotree(T)$.
However, we can find a proof tree $T'$ for the same fact
such that its treewidth only depends on $\aprogram$.
The key observation is that we can rename apart all constants in $T$
such that two occurrences of a constant remain equal only when the 
validity of the proof tree requires them to be.
This results in a tree where each constant appears only within a
connected region: namely, within the subtree rooted at the rule instance
that first introduced it. 

\begin{lemma}\label{lemma:bounded-treewidth}
  For every program $\aprogram$ there exists a constant $k_\aprogram$
  such that for every proof tree $T = \tuple{V, \troot, \ch, \tfact, \trule}$
  for $f$ over $\aprogram$ and database $\Dnter$,
  there is a database $\Dnter'$
  and a proof tree $T'$ for $f$ over $\aprogram$ and $\Dnter'$
  with treewidth at most $k_\aprogram$.
\end{lemma}

\begin{proof}
  In the following, we write $\vars(\mathcal{E})$ and $\terms(\mathcal{E})$ for the set of all variables
  and all terms, respectively, that appear in an expression $\mathcal{E}$.
  We obtain $T' = \tuple{V, \troot, \ch, \tfact', \trule}$ from $T$,
  where $\tfact'$ is defined recursively for each $v \in V$ as follows:
  \begin{itemize}
  \item If $v = \troot$, then $\tfact'(v) \defeq \tfact(v)$.
  \item If $v$ is the $i$-th child of $w \in V$ and $\trule(w) = \rho$,
        let $\theta \colon \vars(\rhead{\rho}) \to \Clang$ 
        be a substitution such that $\theta(\rhead{\rho}) = \tfact'(w)$.
        Given $\tfact(v) = p(t_1, \ldots, t_k)$ and
        $\rbodyi{i}{\rho} = p(s_1, \ldots, s_k)$,
        we set $\tfact'(v) \defeq p(t_1', \ldots, t_k')$ where
        \begin{equation*}
            t_j' = \begin{cases}
              \theta(s_j) & \text{if } s_j \in \vars(\rhead{\rho}) \text{ or } s_j \in \Clang;\ \text{and} \\
              \tuple{t_j, w} & \text{otherwise.}
          \end{cases}
        \end{equation*} 
  \end{itemize}    
  Since the number of distinct constants
  in any single rule instance is bounded by the number of terms
  appearing in the rules of $\aprogram$, we can construct a tree
  decomposition whose bags correspond to rule instances, each containing
  only the constants mentioned in that instance. 
  Let $D = \tuple{V, E, B}$ where $\set{v, w} \in E$ if $v$ is a child of $w$ in $T$.
  We define $B(\ell) \defeq \terms(\tfact'(\ell)) \cup \set{\ell}$
  for each leaf node $\ell$ in $T$; for all other nodes $v \in V$
  with $\ch(v) = \tuple{v_1, \ldots, v_n}$ we set
  \begin{equation}
      B(v) \defeq \set{v} \cup \set{v_i \mid 1 \leq i \leq n} \cup \terms(\tfact'(v)) \cup \bigcup_{1 \leq i \leq n} \terms(\tfact'(v_i))
  \end{equation}
  Then $D$ defines a tree decomposition for $T$
  whose width depends only on the maximal number of distinct terms
  in a rule of $\aprogram$, giving us the desired bound
  $k_\aprogram$.
  Condition~\ref{def:treewidth-1} of Definition~\ref{def:treewidth} 
  is satisfied because every constant is
  confined to a connected subtree; 
  condition~\ref{def:treewidth-2} holds
  because all constants of each fact are witnessed in the bag of the
  node where the corresponding rule is applied. \qed
\end{proof}

Since MSO model checking over structures of
bounded treewidth can be performed in logarithmic
space~\cite[Theorem~1.2]{ElberfeldJT10}, we thus obtain the following.

\begin{theorem}\label{thm:mso-logspace}
  Let $\tau$ be an MSO proof transformation 
  from program $\aprogram_1$ to $\aprogram_2$.
  For any proof tree $T = \tuple{V, \troot, \ch, \tfact, \trule}$ for fact $f$
  over $\aprogram_1$ and database $\Dnter$, the proof DAG
  $\tau(T)$ can be computed by an $O(\log \sizeof{T})$ space bounded transducer. 
\end{theorem}

\begin{proof}
  If $\tau$ is an MSO proof transformation,
  there exists a $k$-dimensional MSO interpretation and set of components $\mathcal{C}$
  satisfying Definition~\ref{def:mso-proof-tree-transformation}.
  For a proof tree $T$, we evaluate each MSO formula
  over $\msotree(T')$, where
  we obtain $T' = \tuple{V, \troot, \ch, \tfact', \trule}$ 
  as in Lemma~\ref{lemma:bounded-treewidth}.   
  Note that we do not need to store $T'$ as an intermediate step, 
  since we can compute $\tfact'(v)$ for any $v \in V$ on demand in logarithmic space.
  For each term position $j$ in $\tfact(v)$, we determine $t_j'$ by 
  walking upward from $v$ toward the root, 
  following how the variable at position $j$ is propagated 
  through the heads of successive rule instances.
  If the variable appears in the head of the current rule, 
  we continue upward, inheriting the corresponding position 
  in the parent's fact. Otherwise, the current node is the ancestor 
  that introduced this constant, and we return the tagged 
  pair $\tuple{t_j, w}$.

  By Lemma~\ref{lemma:bounded-treewidth}, $\msotree(T')$ has bounded treewidth, so each formula can be evaluated in logarithmic space~\cite[Theorem~1.2]{ElberfeldJT10}.
  To instantiate each formula, we need $k$ counters over the vertex set $V$
  for $\varphi_\text{dom}(\vec{x})$ and $\varphi_\rho(\vec{x})$,
  $2k$ counters over $V$ for $\varphi_\text{child}(\vec{x}, \vec{y})$,
  $k$ counters over $V$, and $1$ counter over the set of all terms occurring in fact labels
  for $\varphi_\text{arg}(\vec{x}, d)$. 
  This yields a proof DAG $\tau(T)$ for $f$ over $\Dnter'$.
  We rename each fresh constant $\tuple{t, w}$ to $t$ to get a proof DAG for $f$ over $\Dnter$. \qed
\end{proof}
While this suggests that there are program transformations that
cannot be reversed using an MSO proof transformation (unless
$\LogSpace = \NLogSpace = \LogCFL$), we find that MSO proof tree
transformations are expressive enough to capture many practical
transformations. We have already seen that Example~\ref{ex:mso-transitivity} can be expressed; we provide
further examples in \fullorconf{Appendix~\ref{appendix:mso}}{the extended technical report}.

\section{Conclusions and Outlook}

We have studied the problem of reversing program transformations at
the level of proof trees and established $\PTime$-hardness for the
DAG-encoded case and $\LogCFL$-hardness for the tree-encoded case.  We
have identified two practical transformation languages: Proof tree
homomorphisms are characterised by uniform containment, and capture
local transformations, such as transforming left-linear proof trees
into bushy proof trees (see Example~\ref{ex:uniform}).  MSO proof
tree transformations handle a strictly larger class of restructurings
(such as transforming from bushy proof trees back to left-linear
proof trees, see Example~\ref{ex:mso-transitivity}), but can still be
evaluated efficiently in logarithmic space on tree-encoded input.

For future work, several directions remain open.  For tree-encoded
inputs, there remains a gap between the \LogCFL lower bound and the
\PTime upper bound. It is unclear whether this can be closed.
Similarly, a characterisation of the program transformations that admit a
reversal by an MSO proof transformation, along the lines of the
characterisation of first-order tree-to-tree transductions by
Bojańczyk and Doumane~\cite{BojanczykD20}, is desirable. On the more
practical side, we are implementing proof transformations in our
Datalog engine Nemo~\cite{Ivliev+:Nemo2024,NemoDemo:ICLP23} and plan
to assess the practical overhead of proof tree reconstruction on
realistic knowledge graphs and optimised rule sets. Lastly, it would
be interesting to study how proof tree transformations generalise to
extensions of Datalog such as existential rules, aggregates, or
stratified negation.

\begin{credits}
    \subsubsection{\ackname}
    This work is supported by 
    Deutsche Forschungsgemeinschaft (DFG, German Research Foundation) in project number 389792660 (TRR 248, \href{https://www.perspicuous-computing.science/}{Center for
    Perspicuous Systems}),
    by the Bundesministerium für Forschung, Technologie und Raumfahrt (BMFTR, Federal
    Ministry of Research, Technology and Space)
    in the \href{https://www.scads.de/}{Center for Scalable Data Analytics and
    Artificial Intelligence} (ScaDS.AI), and in DAAD project 57616814 (\href{https://secai.org/}{SECAI, School of Embedded Composite AI})
    as part of the program Konrad Zuse Schools of Excellence in Artificial Intelligence.

    \subsubsection{Supplemental Material Statement.}
    While the paper is self-contained, we provide further examples of
    proof transformations \fullorconf{in the appendix}{in the appendix
    of the full version, which is publicly available online}.

    \subsubsection{Declaration of use of Generative AI.}
    Claude (Anthropic) was used for proof-reading and preparing the figures.
    It also identified and helped correct an error in the proof of Theorem~\ref{thm:hardness-tree}.

\end{credits}

\bibliographystyle{splncs04}
\bibliography{references}

\iffullversion
\newpage
\appendix

\section{Examples of Proof Tree Homomorphisms}\label{appendix:uniform}

In this section, we give further examples of program transformations
that can be reversed using proof tree homomorphisms.
Let $\aprogram_1 = \tuple{R_1, \Pinput, \Poutput}$ be a program.
We consider several common classes of program transformations
that produce a program $\aprogram_2 = \tuple{R_2, \Pinput, \Poutput}$
with $\aprogram_2 \contained \aprogram_1$,
and provide proof tree homomorphisms from $\aprogram_2$ to $\aprogram_1$.

\subsection{Atom Permutation}

A simple program transformation swaps atoms 
within rule bodies~\cite{PettorossiP94}. 
This has no effect on the derived facts 
but may affect evaluation order in a reasoner.
Formally, an \newterm{atom permutation} is a family
$\{\pi_\rho\}_{\rho \in R_1}$, where $\pi_\rho \colon [n] \to [n]$ 
is a permutation with $\sizeof{\rho} = n$ for each $\rho \in R_1$, and $[n] \defeq \set{1, \dotsc, n}$.
For a rule $\rho = H \leftarrow B_1, \ldots, B_n$,
we define 
\begin{equation*}
    \pi_\rho(\rho) \defeq H \leftarrow B_{\pi_\rho^{-1}(1)}, \ldots, B_{\pi_\rho^{-1}(n)}.
\end{equation*}
The transformed program is $\aprogram_2$ with $R_2 = \set{\pi_\rho(\rho) \mid \rho \in R_1}$.
To reverse this transformation, 
we define, for each $\pi_\rho(\rho) \in R_2$,
the proof tree template $S_{\pi_\rho(\rho)}$
for the rule $\pi_\rho(\rho)$ over $\aprogram_1$:

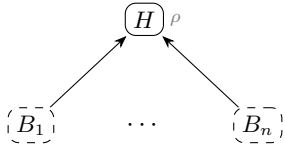
\begin{center}
\begin{tikzpicture}[>=Stealth,
    nd/.style={draw, rounded corners, font=\footnotesize, inner sep=3pt},
    lf/.style={draw, rounded corners, font=\footnotesize, inner sep=3pt, dashed},
    rl/.style={font=\scriptsize, text=black!50, anchor=west, inner sep=1pt, xshift=1pt},
  ]
  \node[lf] (P1) at (0,0) {$B_1$};
  \node at (1.5,0) {$\cdots$};
  \node[lf] (Pn) at (3,0) {$B_n$};
  \node[nd] (R) at (1.5,1.4) {$H$};
  \node[rl] at (R.east) {$\rho$};
  \draw[->] (P1) -- (R);
  \draw[->] (Pn) -- (R);
\end{tikzpicture}
\end{center}

\subsection{Rule Unfolding}

Rule unfolding replaces a body atom with the body 
of a rule that derives it,
eliminating one level of indirection 
in the derivation~\cite{PettorossiP94}.

For a rule $\rho$ and a substitution $\sigma$, we write $\sigma(\rho)$
for the rule obtained by applying $\sigma$ to each atom.
Let $r = H \leftarrow B_1, \ldots, B_n \in R_1$ be a rule 
and $B_k$ be a body atom of $r$ for some $1 \leq k \leq n$.  
Let $d = A \leftarrow C_1, \ldots, C_p \in R_1$ be another rule (with $d \neq r$) 
such that the head of $d$ unifies with $B_k$ 
via a most general unifier~$\sigma$.
The \newterm{unfolding} of $r$ at position $k$ using $d$ is the rule
\begin{align}
    \opfont{unfold}_{r, k}(d) 
    \defeq \sigma\bigl(H \leftarrow B_1, \ldots, B_{k-1}, 
    C_1, \ldots, C_p, B_{k+1}, \ldots, B_n\bigr)
\end{align}
A transformation from $\aprogram_1$ to $\aprogram_2$ is a \newterm{rule unfolding}
if for every $\rho \in R_2$, either $\rho \in R_1$ holds,
or there are $r, d \in R_1$ and $1 \leq k \leq \sizeof{r}$
such that $\rho = \opfont{unfold}_{r, k}(d)$.

To reverse unfolding, consider an unfolded rule 
$\rho' = \opfont{unfold}_{r, k}(d)$.
The template $S_{\rho'}$ reconstructs the original 
two-step derivation:

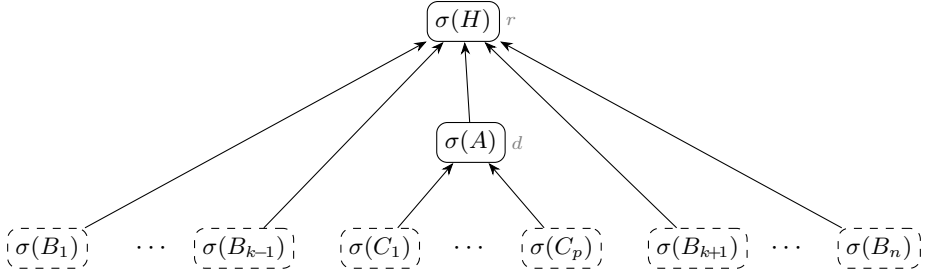
\begin{center}
\begin{tikzpicture}[>=Stealth,
    nd/.style={draw, rounded corners, font=\footnotesize, inner sep=3pt},
    lf/.style={draw, rounded corners, font=\footnotesize, inner sep=3pt, dashed},
    rl/.style={font=\scriptsize, text=black!50, anchor=west, inner sep=1pt, xshift=1pt},
  ]
  \node[lf] (B1) at (0,0) {$\sigma(B_1)$};
  \node at (1.4,0) {$\cdots$};
  \node[lf] (Bk1) at (2.6,0) {$\sigma(B_{k\!-\!1})$};
  
  \node[lf] (C1) at (4.4,0) {$\sigma(C_1)$};
  \node at (5.6,0) {$\cdots$};
  \node[lf] (Cp) at (6.8,0) {$\sigma(C_p)$};
  
  \node[lf] (Bk2) at (8.6,0) {$\sigma(B_{k\!+\!1})$};
  \node at (9.8,0) {$\cdots$};
  \node[lf] (Bn) at (11,0) {$\sigma(B_n)$};
  
  \node[nd] (D) at (5.6,1.4) {$\sigma(A)$};
  \node[rl] at (D.east) {$d$};
  
  \node[nd] (R) at (5.5,3) {$\sigma(H)$};
  \node[rl] at (R.east) {$r$};
  
  \draw[->] (C1) -- (D);
  \draw[->] (Cp) -- (D);
  
  \draw[->] (B1) -- (R);
  \draw[->] (Bk1) -- (R);
  \draw[->] (D) -- (R);
  \draw[->] (Bk2) -- (R);
  \draw[->] (Bn) -- (R);
\end{tikzpicture}
\end{center}

For rules $\rho \in R_2 \cap R_1$ that were not unfolded,
the template $S_\rho$ is the identity:
a root node deriving the head of $\rho$ connected to the leaf nodes 
that are labelled by the body atoms of $\rho$.

\newpage
\section{Examples of MSO Proof Transformations}\label{appendix:mso}

In this section, we present examples of program transformations
that can be reversed using MSO proof transformations.
We let $\aprogram_1 = \tuple{R_1, \Pinput, \Poutput}$ be a program
and define program transformations obtaining $\aprogram_2 = \tuple{R_2, \Pinput, \Poutput}$,
which can be reversed by proof transformations from $\aprogram_2$ to $\aprogram_1$.

\subsection{Rule Folding}

Rule folding replaces multiple rules 
whose bodies share a common pattern 
by a single rule that uses an auxiliary 
predicate~\cite{PettorossiP94}.
Let $c_1, \ldots, c_n$ and $d_1, \ldots, d_n$ be pairwise distinct rules in $R_1$,
and let $A$ and $H$ be atoms and $\bar{B}$ a list of atoms such that:
\begin{itemize}
    \item $\rhead{d_i}$ is unifiable with $A$ 
    using most general unifier $\sigma_i$,
    \item $c_i = \sigma_i(H \leftarrow \bar{B},\, \rbody(d_i))$, and
    \item the head of every rule $r \notin \{d_1, \ldots, d_n\}$ 
    is not unifiable with $A$.
\end{itemize}
The \newterm{folding} of $c_1, \ldots, c_n$ using $d_1, \ldots, d_n$
is the rule $r \defeq H \leftarrow A,\, \bar{B}$.
The transformed program $\aprogram_2$ is obtained from $\aprogram_1$
by replacing $\{c_1, \ldots, c_n\}$ with $\{r\}$.

We can reverse this transformation using a $1$-dimensional MSO proof transformation
with a single component.
The idea is to replace the rule $r$ with rule $c_i$ if the $A$ was derived by rule $d_i$.
\begin{align}
    \varphi_\text{dom}(x) &\defeq 
    \bigwedge_{i=1}^{n} \neg\, \msorule_{d_i}(x) \\
    \varphi_{c_i}(x) &\defeq 
    \msorule_r(x) \land 
    \exists w.\; \msochild_1(w, x) \land \msorule_{d_i}(w) \\
    \varphi_\rho(x) &\defeq \msorule_\rho(x) 
    \qquad \text{for } \rho \notin \set{c_1, \ldots, c_n}
\end{align}
The child formula rewires the children of an $r$-node,
placing the $\bar{B}$-children first 
and then the children inherited from the absorbed $d_i$-node:
\begin{align}
  \begin{split}
    \varphi^j_\text{child}(x, y) &\defeq
    \big(\neg\, \msorule_r(y) \land \msochild_j(x, y)\big) \\
    & \lor
    \big(\msorule_r(y) \land \msochild_{j+1}(x, y)\big)
    \qquad \text{if } 1 \leq j \leq |\bar{B}|
  \end{split} \\
  \begin{split}
    \varphi^j_\text{child}(x, y) &\defeq
    \big(\neg\, \msorule_r(y) \land \msochild_j(x, y)\big) \\
    &\lor
    \big(\msorule_r(y) \land 
    \exists w.\; \msochild_1(w, y) 
    \land \msochild_{j - |\bar{B}|}(x, w)\big)
    \qquad \text{if } j > |\bar{B}|
    \end{split} \\
    \varphi^j_\text{arg}(x, d) &\defeq \msoarg_j(x, d)
\end{align}

\subsection{Filter Propagation}

Filter propagation is a technique that moves filter conditions 
such as $x = 3$ through the rules of a program,
so that they are checked as early as possible 
during bottom-up evaluation~\cite{HK2026}.

Formally, a \newterm{filter propagation} is a mapping $f$ 
that assigns to each body atom of each rule in $R_1$
a new rule and body position in which it should appear.
That is, $f$ maps atom positions $\tuple{\rho, i}$ 
with $\rho \in R_1$ and $1 \leq i \leq |\rho|$
to tuples $\tuple{\rho', j, \sigma}$,
where $\rho' \in R_1$, $j \geq 1$
and $\sigma$ is a substitution adapting the variables.
Atom positions with $f(\rho, i) = \tuple{\rho, i, \mathit{id}}$ are \newterm{stationary};
all others are \newterm{moved}.
The transformed program $\aprogram_2 = \tuple{R_2, \Pinput, \Poutput}$
has, for each $\rho' \in R_1$, a rule $f(\rho')$ 
whose head is $\rhead{\rho'}$
and whose $j$-th body atom is $\sigma(\rbodyi{i}{\rho})$
for the unique $\tuple{\rho, i}$ with $f(\rho, i) = \tuple{\rho', j, \sigma}$.
We set $R_2 \defeq \set{f(\rho) \mid \rho \in R_1}$.

\begin{example}
Consider the program $\aprogram_1$:
\begin{align}
    \rho_1 &:\ r(x, y) \leftarrow e(x, y) \\
    \rho_2 &:\ r(x, y) \leftarrow r(x, z), e(z, y) \\
    \rho_3 &:\ \predname{out}(y) \leftarrow r(x, y), x = a
\end{align}
The filter $x = a$ at position $(\rho_3, 2)$ is propagated
to the base case, yielding:
\begin{align*}
    f(\rho_1) &:\ r(x, y) \leftarrow e(x, y), x = a \\
    f(\rho_2) &:\ r(x, y) \leftarrow r(x, z), e(z, y) \\
    f(\rho_3) &:\ \predname{out}(y) \leftarrow r(x, y)
\end{align*}
Here $f(\rho_3, 2) = \tuple{\rho_1, 2, \mathit{id}}$
and all other positions are stationary.
In a proof tree for $\aprogram_2$,
the filter appears at every base-case leaf;
in a proof tree for $\aprogram_1$,
it appears only at the root.
\end{example}

An MSO proof transformation from $\aprogram_2$ to $\aprogram_1$ 
can be defined using a $1$-dimensional MSO interpretation 
with a single component.
\begin{align}
    \varphi_\text{dom}(x) &\defeq \top \\
    \varphi_\rho(x) &\defeq \msorule_{f(\rho)}(x) \\
    \begin{split}
    \varphi^i_\text{child}(x, y) &\defeq 
    \bigvee_{\substack{f(\tuple{\rho, i}) = \tuple{\rho', j, \sigma}}}
    \msorule_{f(\rho)}(y) \land
    \exists p.\; \msochild_j(x, p) \\
    &\land \msorule_{f(\rho')}(p) 
    \land p \leq y
    \end{split} \\
    \varphi^j_\text{arg}(x, d) &\defeq \msoarg_j(x, d)
\end{align}
The domain formula retains all nodes
and the rule formula recovers the original rule name.
The argument formula is inherited directly,
since filter propagation does not change rule heads.
The child formula 
links a filter that has been pushed down in the tree to its original ancestor.

\subsection{Projection Pushing}

Projection pushing removes argument positions in predicates
such that the output of the program is unaffected~\cite{RamakrishnanBK88}.

Formally, let $\aprogram_1 = \tuple{R_1, \Pinput, \Poutput}$ be a program.
A \newterm{position elimination} assigns to each IDB predicate $p$ 
a set of \newterm{retained positions} 
$U(p) \subseteq \{1, \ldots, \arity(p)\}$
such that all positions of output and input predicates are retained.
For an atom $p(t_1, \ldots, t_m)$, 
we write $p|_U(t_1, \ldots, t_m)$ 
for the atom $p'(t_{j_1}, \ldots, t_{j_k})$ 
where $U(p) = \{j_1 < \cdots < j_k\}$ 
and $p'$ is a fresh predicate of arity $k$.
For a rule $\rho = H \leftarrow B_1, \ldots, B_n$,
the \newterm{reduced rule} is
\begin{equation}
    \rho|_U \defeq H|_U \leftarrow B_1|_U, \ldots, B_n|_U
\end{equation}
provided $\rho|_U$ is safe. 
The transformed program is 
$\aprogram_2 = \tuple{R_2, \Pinput, \Poutput}$ 
with rule set $R_2 = \set{\rho|_U \mid \rho \in R_1}$.

\begin{example}
Consider the program $\aprogram_1$ computing the transitive closure of $e$:
\begin{align*}
    r(x, y) &\leftarrow e(x, y) \\
    r(x, y) &\leftarrow r(x, z), e(z, y) \\
    \predname{out}(y) &\leftarrow r(x, y)
\end{align*}
If we are only interested in the output predicate $\predname{out}$,
the first position of $r$ can be eliminated.
Setting $U(r) = \set{2}$ yields $\aprogram_2$:
\begin{align*}
    r'(y) &\leftarrow e(x, y) \\
    r'(y) &\leftarrow r'(z),\, e(z, y) \\
    \predname{out}(y) &\leftarrow r'(y)
\end{align*}
\end{example}

This transformation can be reversed using 
a $1$-dimensional MSO interpretation 
with a single component.
The tree structure is preserved 
and the rule labelling recovers the original rule from the reduced one
for each $\rho \in R_1$:
\begin{align}
  \varphi_\text{dom}(x) &\defeq \top \\
  \varphi^i_\text{child}(x, y) &\defeq \msochild_i(x, y) \\
  \varphi_{\rho}(x) &\defeq \msorule_{\rho|_U}(x)
\end{align}

Defining the argument formula is a bit more involved.
For retained positions, the argument value is read directly 
from the reduced representation.
For eliminated positions, we trace the value downward through the tree.
Fix a rule $\rho \in R_1$ with head predicate $p$ 
and let $j \notin U(p)$ be an eliminated position.
Let $x_j$ denote the variable at position $j$ in $\rhead{\rho}$.
By safety, $x_j$ occurs in at least one body atom of $\rho$.
For each body atom $\rbodyi{i}{\rho}$ with predicate $q$
in which $x_j$ occurs at position $j'$, 
exactly one of the following holds:
\begin{enumerate}
    \item $q$ is an input predicate, or $j' \in U(q)$:
    the value can be read at the $i$-th child of $x$.
    \item $q$ is an IDB predicate and $j' \notin U(q)$:
    position $j'$ is itself eliminated, 
    and we must recurse at the $i$-th child.
\end{enumerate}

To encode this recursive computation, we define a fixed finite graph
that captures how variable bindings propagate through rules.
The \newterm{predicate position graph} of $\aprogram_1$
is the labelled directed graph $G_\aprogram$ whose vertices are
predicate-position pairs $\tuple{p, j}$
with $p$ a predicate and $1 \leq j \leq \arity(p)$.
It contains an edge $\tuple{p, j} \to \tuple{q, j'}$ labelled $\tuple{\rho, i}$
whenever $\rho \in R_1$ has head predicate $p$,
the variable at position $j$ in $\rhead{\rho}$
also appears at position $j'$ in $\rbodyi{i}{\rho}$,
and that body atom has predicate $q$.
A vertex $\tuple{q, j'}$ is \newterm{terminal}
if $q \in \Pinput$ or $j' \in U(q)$.

The MSO formula $\Phi_{p,j}(x, d)$
quantifies over a downward path in the proof tree
together with a run of $G_\aprogram$ along that path.
For each vertex $\tuple{q, j'}$ of $G_\aprogram$,
we introduce a monadic variable $X_{q,j'}$
representing the set of tree nodes on the path
that are visited in graph node $\tuple{q, j'}$.
Let $V(G_\aprogram)$ be the set of vertices in $G_\aprogram$
and let $\bar X_{q, j'}$ be a list of monadic second-order variables
$X_{q, j'}$ for each vertex $\tuple{q, j'}$ in $G_\aprogram$.

\begin{align}
  \begin{split}
    \Phi_{p,j}(x, d) &\defeq
    \exists \bar X_{q,j'}.\;\exists y.\; \Biggl(
    x \in X_{p,j} \land \forall v.\; \biggl( \\
    &\bigwedge_{\tuple{q, j'} \in V(G_\aprogram)} \biggl(
    v \in X_{q,j'} \Rightarrow \\
    &\qquad \bigvee_{\substack{\tuple{q, j'} \to \tuple{q', j''} \\ \text{labelled } \tuple{r, i}}} \biggl(
    \msorule_{r|_U}(v) \land
    \exists w.\; \msochild_i(w, v) \\ \\
    &\qquad\qquad \land
    \begin{cases}
        w = y \land \msoarg_{j''}(y, d)
        & \text{if } \tuple{q', j''} \text{ terminal} \\
        w \in X_{q', j''}
        & \text{otherwise}
      \end{cases}
      \biggr)\biggr)\biggr)\Biggr)
    \end{split}
\end{align}
At the start, node $x$ is in state $\tuple{p, j}$.
At each step, a node $v$ in state $\tuple{q, j'}$
picks an edge $\tuple{q, j'} \to \tuple{q', j''}$ labelled $\tuple{\rho, i}$,
verifies that $v$ uses rule $\rho|_U$,
and descends to child $i$.
If $\tuple{q', j''}$ is terminal,
the value $d$ is read at position $j''$ of that child.
Otherwise, the child enters state $\tuple{q', j''}$
and the process continues.

The argument formula distinguishes retained and eliminated positions:
\begin{align}
    \varphi^j_\text{arg}(x, d) &\defeq
    \bigvee_{\substack{r \in R_1 \\ \text{head pred.\ } p}}
    \msorule_{r|_U}(x) \;\land\;
    \begin{cases}
        \msoarg_{j'}(x, d)
        & \text{if } j \in U(p) \\
        \Phi_{p,j}(x, d)
        & \text{if } j \notin U(p)
    \end{cases}
\end{align}
where $j' = |\{k \in U(p) \mid k \leq j\}|$ 
is the position of $j$ in the reduced predicate.
For retained positions, the value is read directly 
from the reduced representation 
(adjusting for the change in indexing).
For eliminated positions, the value is recovered via $\Phi_{p,j}$.

\fi

\end{document}